\documentclass[letterpaper, 11pt]{amsart}

\usepackage{framed, multirow}
\usepackage[linesnumbered,ruled]{algorithm2e}
\usepackage{mathrsfs}
\usepackage{amssymb}
\usepackage{amsthm}
\usepackage{amsmath}
\usepackage{amsfonts}
\usepackage{latexsym}
\usepackage{hyperref}
\usepackage{multicol}
\usepackage{graphicx}
\usepackage[margin=1.07in]{geometry}
\usepackage{float}

\usepackage{url}
\usepackage{xcolor}
\definecolor{newcolor}{rgb}{.8,.349,.1}

\newtheorem{theorem}{Theorem}
\newtheorem{corollary}[theorem]{Corollary}
\newtheorem{lemma}[theorem]{Lemma}
\newtheorem{definition}[theorem]{Definition}

\newtheorem{remark}[theorem]{Remark}
\newtheorem{example}[theorem]{Example}

\author[J. Espinoza]{Jes\'us F. Espinoza}
\email{jesus.espinoza@mat.uson.mx}
\thanks{Corresponding author: jesus.espinoza@mat.uson.mx (Jes\'us F. Espinoza)}

\author[R. Hern\'andez]{Rosal\'ia Hern\'andez-Amador}
\email{rosalia.hdez@mat.uson.mx}

\author[H. Hern\'andez]{H\'ector A. Hern\'andez}
\email{hectorhdez@gmail.com}

\author[B. Ramonetti]{Beatriz Ramonetti-Valencia}
\email{beatrizramonetti92@gmail.com}

\address{Departamento de Matem\'aticas, Universidad de Sonora, M\'exico.}

\title{A numerical approach for the filtered generalized \v{C}ech complex}

\begin{document}

\begin{abstract}
In this paper, we present an algorithm to compute the filtered generalized \v{C}ech complex for a finite collection of disks in the plane, which don't necessarily have the same radius.

The key step behind the algorithm is to calculate the minimum scale factor needed to ensure rescaled disks have a nonempty intersection, through a numerical approach, whose convergence is guaranteed by a generalization of the well-known Vietoris-Rips Lemma, which we also prove in an alternative way, using elementary geometric arguments.

We present two applications of our main results. We give an algorithm for computing the 2-dimensional filtered generalized \v{C}ech complex of a finite collection of $d$-dimensional disks in $\mathbb{R}^d$. In addition, we show how the algorithm yields the minimal enclosing ball for a finite set of points in the plane.
\end{abstract}

\maketitle

% MSC2010: 68U05, 65D17, 05E45

\section{Introduction}
Recently, in the study of data point clouds from a topological approach (cf. \cite{Carlsson:Topology-Data, Carlsson:2014, Ghrist:2008, Carlsson-et-al:2013, Zomorodian-Carlsson:2005}) the need to develop algorithms to construct different simplicial structures has arisen, such as the Vietoris-Rips complex, the \v{C}ech complex, the PL-lower star complex, etc. (cf. \cite{edelsbrunner:2010, Zomorodian2010FastCO}).

Of particular interest to us is the generalized \v{C}ech complex structure: whereas the standard \v{C}ech complex is induced by the intersection of a collection of disks with fixed radius, the generalized version admits different radii (see \cite{CechGenelized}); when radiuses are rescaled, using the same scale factor each time, the corresponding simplicial complexes forms the filtered generalized \v{C}ech complex.

There exist efficient algorithms to calculate the standard \v{C}ech complex (e.g. \cite{DANTCHEV2012708}), and software currently available to obtain the associated filtration (cf. \cite{Dionysus:2018, Otter:2017}); also, in \cite{Kerber:2013} the authors propose an algorithm to approximate the \v{C}ech filtration. On the other hand, we can find algorithms to calculate the generalized \v{C}ech complex (e.g. \cite{CechGenelized}), however as far as we know, there are neither algorithms nor software to provide the filtered generalized \v{C}ech complex. In the present work we show an algorithm to compute the filtered generalized \v{C}ech complex for a finite collection of disks, specifically, in the plane. Actually, we also show an algorithm to build up to the 2-dimensional filtered generalized structure (or 2-skeleton), for higher dimensional disk systems, which many applications only require, as we can see in \cite{Bendich:2016, Silva-Ghrist:2007, Goldfarb:2014, Ramamoorthy:2016, Robins:2016}.

The key step behind these proposed algorithms, is to calculate the minimum scale factor (called \textit{\v{C}ech scale}) needed to ensure that the rescaled disks have a nonempty intersection; the generalized Vietoris-Rips Lemma over multiple radii will allow us to calculate these scales numerically. 

We must emphasize that, our main algorithm (Algorithm \ref{Algorithm:Cech-scale-triplets}) is only generalizable to higher dimensional disk systems to obtain the 2-dimensional filtered generalized \v{C}ech structure, as we show as an application. Additionally, we show how our algorithm yields the minimal enclosing ball for a finite set of points in the plane.

This paper is organized as follows. 
In Section \ref{Section:Rips.Cech.systems} we introduce basic notions and notation which will be used throughout the paper. We define the Vietoris-Rips system and the \v{C}ech system, associated to a finite collection of closed disks in the euclidean space (or \textit{disk system}) in terms of their intersection. We also introduce the fundamental notions of Vietoris-Rips scale and \v{C}ech scale for a disk system, as the infimum over all rescaling factors such that the disk system becomes a Vietoris-Rips system or a \v{C}ech system respectively. In Lemma \ref{Lemma:VR-generalized} we state and prove a generalization, over multiple radii, of the well-known Vietoris-Rips Lemma \cite[Th. 2.5]{Silva-Ghrist:2007} using elementary geometric arguments. In \cite{MartinTesis:2017} can be found a proof in the generalized case, following the ideas in \cite{Silva-Ghrist:2007}.

In Section \ref{Section:filtered-simplicial-structures} we describe the generalized versions of standard Vietoris-Rips and \v{C}ech simplicial complex structures, to the case of disk systems with different radiuses. We explain how their respective filtrations are induced by weight functions, and we propose an algorithm to obtain the \v{C}ech-weight function of a given disk system, associating to each \v{C}ech simplex its corresponding \v{C}ech scale.

Section \ref{Section:intersection-properties} focuses on studying the intersection properties of collections of disks in the plane. We define a real-valuated function associated to each disk system in the plane, such that, if it turns out to be nonnegative, then its \v{C}ech scale agrees with its Vietoris-Rips scale, being then easy to compute; otherwise, the \v{C}ech scale will correspond to a root of such function, and  we propose a numerical approach to obtain this \v{C}ech scale (Section \ref{Section:Cech.scale-algorithm}), supported on the generalized Vietoris-Rips Lemma which provides appropriated bounds.

Section \ref{Section:Cech.scale-algorithm} contains our main result, the \texttt{Cech.scale} algorithm, whose input is a disk system in the plane, and the output is the corresponding \v{C}ech scale, as well as the unique intersection point of the rescaled disk system at its \v{C}ech scale (see Lemma \ref{Lemma:unique-point}).

Finally, we conclude the paper illustrating two applications in Section \ref{Section:applications}. First, we recall the miniball problem, to show how our \texttt{Cech.scale} algorithm yields the minimal enclosing ball for a finite point cloud in the plane. Secondly, we present an algorithm for computing the \v{C}ech filtration of the 2-skeleton of the generalized \v{C}ech complex structure for a $d$-dimensional disk systems in an arbitrary euclidean space $\mathbb{R}^d$.

\section{Vietoris-Rips and \v{C}ech systems} \label{Section:Rips.Cech.systems}

Throughout this paper, a finite collection of closed $d$-disks in the euclidean space $\mathbb{R}^d$, with positive radius,
\begin{equation}\label{generalDisksSystem}
M=\{D_i(c_i;r_i) \subset \mathbb{R}^d \mid r_i >0, 1\leq i \leq m \}
\end{equation}
will be called \textit{$d$-disk system}, or simply \textit{disk system} when there is no risk of confusion. In this section we introduce and analyze two fundamental subclasses of disk systems, namely, the Vietoris-Rips systems and the \v{C}ech systems. We study the infimum of those scales that turn a disk system into a Vietoris-Rips or \v{C}ech system.
We conclude this section presenting a generalized version of the Vietoris-Rips Lemma, extended to disk systems.
\begin{definition} Let $M=\{D_1,D_2,\ldots , D_m \}$ be a disk system. We say $M$ is a \textit{Vietoris-Rips system} if $D_i\cap D_j \neq \emptyset$ for each pair $i,j\in \{1,2,\ldots, m\}$. 
Moreover, if the disk system $M$ has the nonempty intersection property $\bigcap_{D_i \in M} D_i \neq \emptyset$, then $M$ is called a \textit{\v{C}ech system}.
\end{definition}

For each $\lambda \geq 0$, and disk system $M$ as in (\ref{generalDisksSystem}) we define the collection 
\[ M_\lambda := \{D_i(c_i;\lambda r_i) \subset \mathbb{R}^d \mid D_i \in M\} \]
and say that $\lambda$ is a \textit{scale}.
Geometrically, the set $M_\lambda$ consists of disks with the same centers than those in $M$, but with rescaled radii by $\lambda$. Clearly, only when $\lambda > 0$ the set $M_\lambda$ will be again a disk system. Note that $M_1=M$, and $M_0$ is the set consisting of the centers of the disks in $M$. 

\begin{definition} \label{Definition:VR-Cech-scales}
Let $M$ be a disk system. The Vietoris-Rips scale of $M$ is defined by,
\[ \nu_M := \inf \{ \lambda \in \mathbb{R} \mid M_\lambda \mbox{ is a Vietoris-Rips system} \}.\]
Analogously, the \v{C}ech scale of $M$ is defined by,
\[ \mu_M := \inf \{ \lambda \in \mathbb{R} \mid M_\lambda \mbox{ is a \v{C}ech system} \}.\]
\end{definition}

Let $\mu_M$ be the \v{C}ech scale of the disk system $M$, then we have that $\bigcap_{D_i \in M} D_i(c_i; \mu_M r_i) \neq \emptyset$. Essentialy, this is a consequence of the completeness of the euclidean space, the fact that $\bigcap_{D_i \in M} D_i(c_i; \lambda r_i) \subset \mathbb{R}^d$ is a closed subset for every scale $\lambda$ in the set $\{ \lambda \in \mathbb{R} \mid M_\lambda \mbox{ is a \v{C}ech system} \}$ and  $\bigcap_{D_i \in M} D_i(c_i; \lambda' r_i) \subset \bigcap_{D_i \in M} D_i(c_i; \lambda r_i)$ for $\lambda' < \lambda$.

A straightforward calculation shows the following characterizations: $M$ is a Vietoris-Rips system if and only if, $\nu_M \leq 1$ (in particular, $\nu_{M_{\nu_M}} = 1$); similarly, $M$ is a \v{C}ech system if and only if, $\mu_M \leq 1$.

Note that it is not hard to calculate the Vietoris-Rips scale $\nu_M$ for a given disk system $M=\{D_1,D_2,\ldots , D_m \}$: if $r_i$ denotes the radius of $D_i$, and $\Vert c_i - c_j \Vert$ represents the distance between the center of $D_i$ and $D_j$, then $\nu_M = \max_{i<j}\{\Vert c_i - c_j \Vert/(r_i + r_j)\}$.

For a disk system $M$ with just one disk, its Vietoris-Rips scale is $\nu_M = 0$; if $M = \{D_1, D_2\}$ has two disks, then $\nu_M = \Vert c_1 - c_2 \Vert/(r_1 + r_2)$. Actually, in both cases the Vietoris-Rips scale agrees with the \v{C}ech scale.

On the other hand, calculating the \v{C}ech scale is a more complicated issue if the disk system has at least three disks. Concerning to \v{C}ech scales, we have the following lemma, which will become important for our implementations.

\begin{lemma} \label{Lemma:unique-point}
Let $\mu \geq 0$ be a scale and let $M$ be a disk system. Then $\mu$ is the \v{C}ech scale of $M$ if and only if, the $\mu$-rescaled system $M_{\mu}$ has only one intersection point, i.e. the set $\bigcap_{D_i \in M} D_i(c_i; \mu r_i)$ is unitary.
\end{lemma}

Such point in $\bigcap_{D_i \in M} D_i(c_i; \mu_M r_i)$ will be denoted by $c_M$.

\begin{proof} 
The case $\mu = 0$ happens only when the disk system consists of a single disk or is a collection of concentric disks. In this case, the claim of the lemma is evident.

Let $\mu > 0$ be the \v{C}ech scale of $M$ and suppose there exist a couple of points $p_1, p_2 \in \bigcap_{D_i \in M} D_i(c_i; \mu r_i)$ such that $p_1 \neq p_2$. By convexity of the disks, it follows that the middle point $\bar{p} = \frac{1}{2}(p_1 + p_2)$ must belong to every disk $D_i(c_i; \mu r_i)$. On the other hand, $\Vert \bar{p} - c_i \Vert < \max \{ \Vert p_1 - c_i \Vert, \Vert p_2 - c_i \Vert  \}$ for any center $c_i$ in the disk system. Let $\mu_i < \mu$ be a scale such that $\bar{p} \in D_i(c_i; \mu_i r_i)$ for every disk in $M$. It follows that $\bar{\mu} = \max \{\mu_i\} < \mu$ and $\bar{p} \in \bigcap_{D_i \in M} D_i(c_i; \bar{\mu}r_i)$ which contradicts the minimality of the \v{C}ech scale $\mu$. Therefore, the set $\bigcap_{D_i \in M} D_i(c_i; \mu r_i)$ is unitary.

Now, suppose $\bigcap_{D_i \in M} D_i(c_i; \mu r_i)$ is unitary and consider the set $S=\{ \lambda \in \mathbb{R} \mid M_\lambda \mbox{ is a \v{C}ech system} \}$. 

If $\bigcap_{D_i \in M} D_i(c_i; \mu r_i) = \{p\}$, then $p \in \partial D_{i_0}(c_{i_0}; \mu r_{i_0})$ for some $D_{i_0} \in M$, because otherwise there would exist a neighborhood of $p$ entirely contained in $\bigcap_{D_i \in M} D_i(c_i; \mu r_i)$.

The fact that $\mu = \inf S$ is a consequence of the following:
\begin{itemize}
\item Let $\lambda \in S$ be a scale such that $\bigcap_{D_i \in M} D_i(c_i; \lambda r_i) \neq \emptyset$. \\
If $ \lambda < \mu$, then $D_i(c_i; \lambda r_i) \subset D_i(c_i; \mu r_i)$ for any $D_i \in M$, and $p \not\in D_{i_0}(c_{i_0}; \lambda r_{i_0})$; thus, $\bigcap_{D_i \in M} D_i(c_i; \lambda r_i) \subset \bigcap_{D_i \in M} D_i(c_i; \mu r_i) = \{p\}$ and $p \not\in \bigcap_{D_i \in M} D_i(c_i; \lambda r_i)$, therefore $\bigcap_{D_i \in M} D_i(c_i; \lambda r_i) = \emptyset$. Then $\lambda \geq \mu$, for any $\lambda \in S$.
\item For every $\varepsilon > 0$, we have $\{p\} \in D_i(c_i; \mu r_i) \subset D_i(c_i; (\mu + \varepsilon)r_i)$. Then $\bigcap_{D_i \in M} D_i(c_i; (\mu + \varepsilon)r_i) \neq \emptyset$, i.e. $\mu + \varepsilon \in S$.
\end{itemize}
\end{proof}

\begin{example}\label{Example:VR-Cech-system}\normalfont
Figure \ref{Figure:Vietoris-Rips-system} shows (left picture) the following 2-disk system in $\mathbb R^2$,
\[ M= \{D_1((-3,4);4), D_2((1,3);3), D_3((2,-1);2) \}. \]

\begin{figure*}[hbt]
\begin{multicols}{3}
\includegraphics[width=0.33\textwidth]{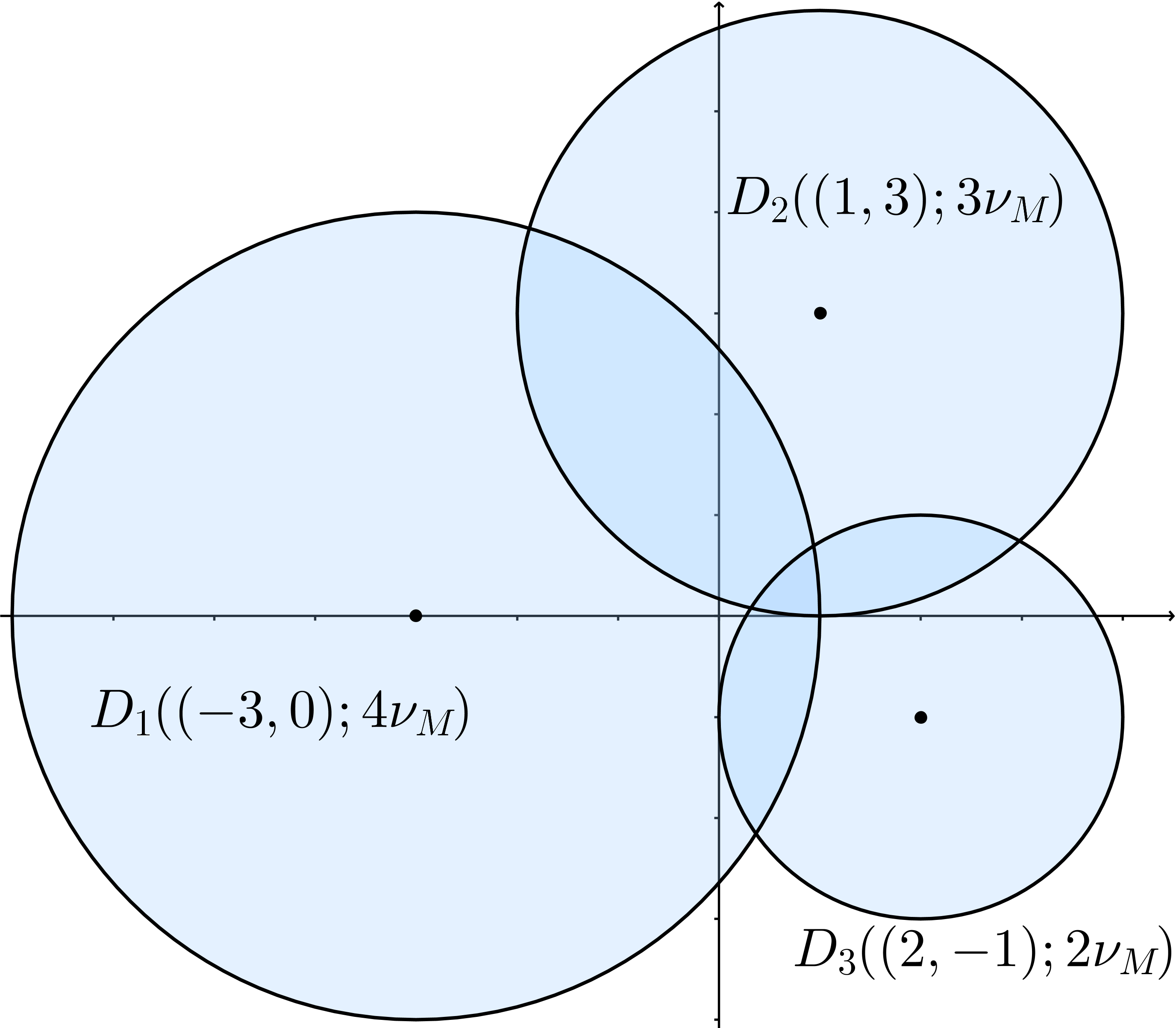}

\includegraphics[width=0.33\textwidth]{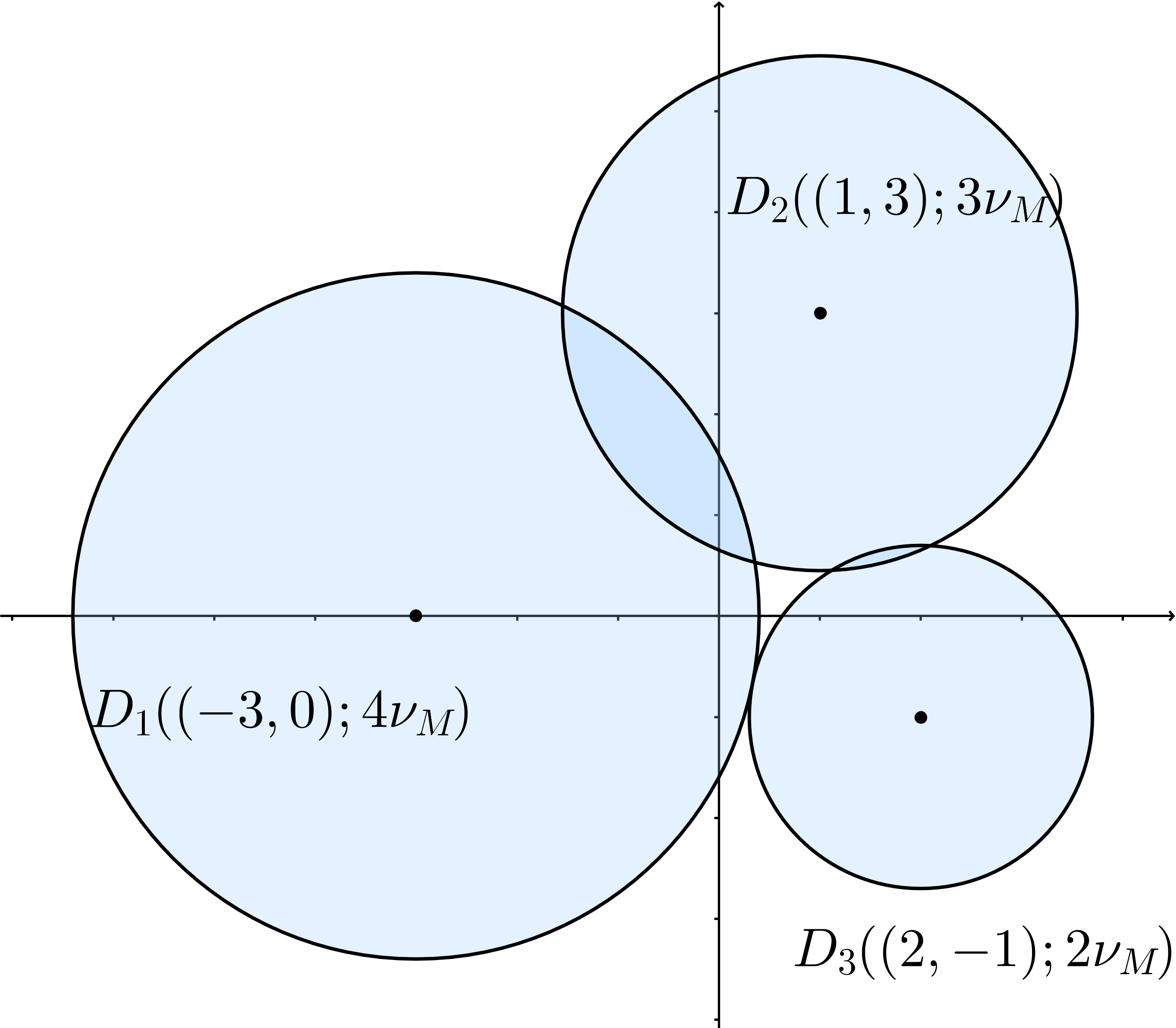}

\includegraphics[width=0.33\textwidth]{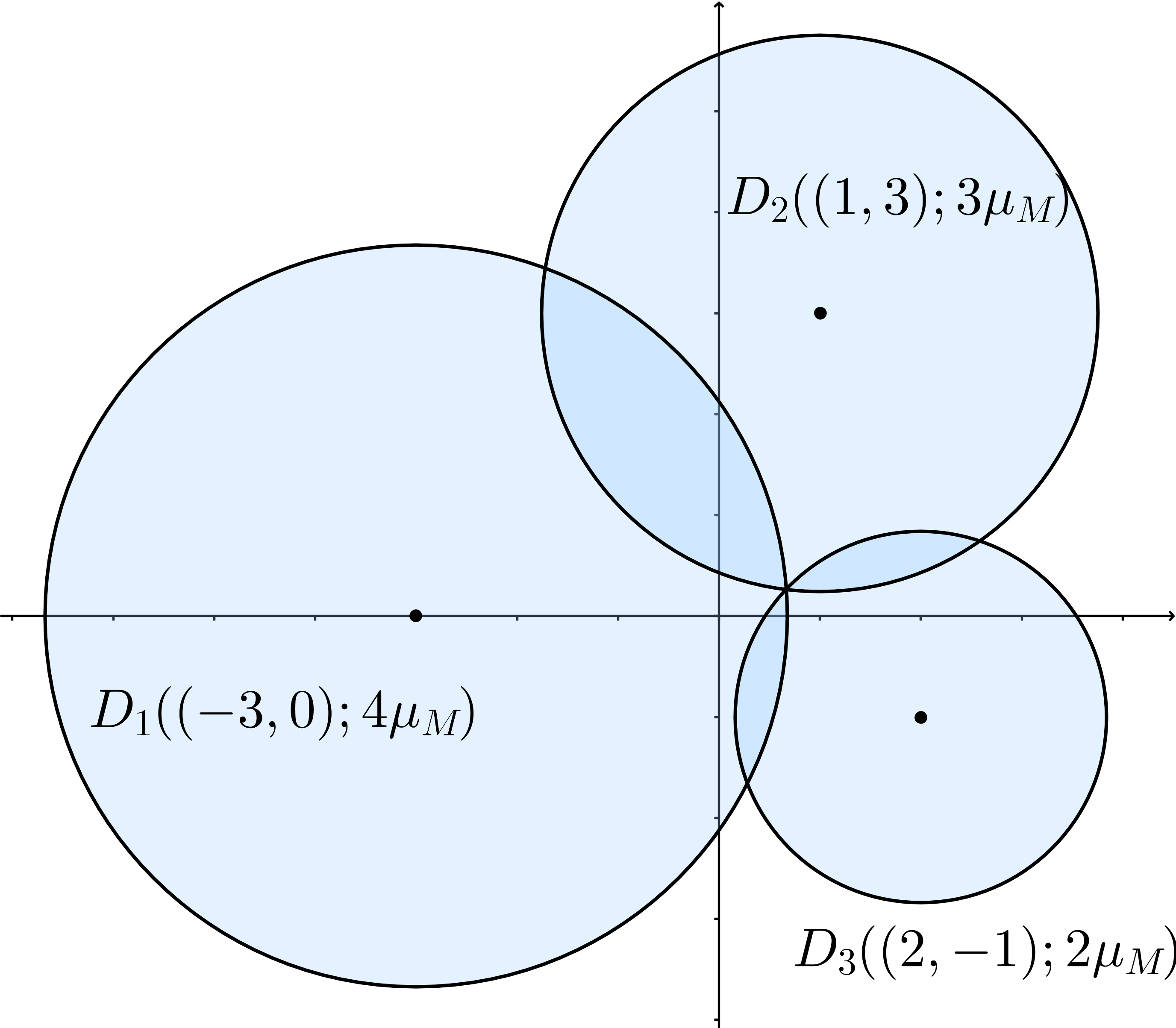}
\end{multicols}
\caption{Vietoris-Rips and \v{C}ech systems in the plane.}
\label{Figure:Vietoris-Rips-system}
\end{figure*}

This 2-disk system is a Vietoris-Rips system and also a \v{C}ech system. In this case, we have $\nu_M = \sqrt{26}/6 \approx 0.8498$, and actually, the $\nu_M$-rescaled 2-disk system $M_{\nu_M}$ (center picture), has an empty intersection, i.e. $\bigcap_{D_i \in M} D_i(c_i; \nu_M\, r_i) = \emptyset$, so it corresponds to a Vietoris Rips system which is not a \v{C}ech system. For such 2-disk system, the \texttt{Cech.scale} algorithm in Section \ref{Section:Cech.scale-algorithm} yields to $\mu_M = 0.9188$ and the \v{C}ech system $M_{\mu_M}$ is shown in the right picture.
\end{example}

There exists a close relationship between Vietoris-Rips systems and \v{C}ech systems. Obviously, every \v{C}ech system is also a Vietoris-Rips system, but the opposite statement does not hold  in general (as we saw in the example above). However, we have the following result which extends the standard Vietoris-Rips Lemma \cite[Th. 2.5]{Silva-Ghrist:2007}, to any Vietoris-Rips system (compare with \cite[Th. 3.2]{MartinTesis:2017}). 

The standard Vietoris-Rips Lemma is established in the context of a disk system in which all radii are equal, and corresponds to a reformulation of the well-known Jung's Lemma \cite{Jung:1901}. On the other hand, for a disk system in general, the Vietoris-Rips Lemma is also valid, although it does not follow directly from Jung's Lemma. Next, we propose a proof of the Vietoris-Rips Lemma, for an arbitrary disk system, using elementary geometrical arguments.

\begin{lemma} \label{Lemma:VR-generalized}
Let $M = \{D_i(c_i; r_i) \subset \mathbb{R}^d \mid r_i >0 \}$ be a finite set of closed disks in $\mathbb{R}^d$. If $D_i(c_i; r_i) \cap D_j(c_j; r_j) \neq \emptyset$ for every pair of disks in $M$, then 
\[ \bigcap_{D_i \in M} D_i(c_i; \sqrt{2d/(d+1)}\, r_i) \neq \emptyset. \]
\end{lemma}

In other words, for every Vietoris-Rips system $M$ in $\mathbb R^d$, the $\sqrt{2d/(d+1)}$-rescaled disk system $M_{\sqrt{2d/(d+1)}}$ is a \v{C}ech system.

\begin{proof}
First, we prove the result for the case where $M$ has at most $d+1$ disks, say $M=\{D_1, \ldots, D_{d'}\}$, $d' \leq d+1$. Note that we need to prove that $\mu_M \leq \sqrt{\frac{2d}{d+1}}$.

Let $\{c_M\} = \bigcap_i D_i(c_i; \mu_M r_i)$ be the unique intersection point of the $\mu_M$-rescaled disk system (see Lemma \ref{Lemma:unique-point}). Without loss of generality, we assume that $\Vert c_i - c_M \Vert = \mu_M r_i$ for $c_1, \ldots, c_m$ and $m \leq d'$.

Then $c_M$ belongs to the convex hull of the set $\{c_1, \ldots, c_m\}$, for if this were not true, there would exist an hyperplane across $c_M$ such that the set $\{c_1, \ldots, c_m\}$ is completely contained at one side, let $v$ be the normal vector for such hyperplane in the opposite direction, then $\langle v, c_i - c_M \rangle > 0$ for all $i=1,\ldots,m$, hence
\[ \Vert c_i - c_M \Vert^2 	= \Vert c_i - (c_M + tv) \Vert^2 + 2t\langle v, c_i - c_M \rangle - t^2\Vert v \Vert^2 > \Vert c_i - (c_M + tv) \Vert^2 \]
for every $t \in I_v := (0, 2 \langle v, c_i - c_M \rangle / \Vert v \Vert^2)$; this implies that $c_M + tv \in D(c_i; \mu_M r_i)$ for any $i=1,\ldots,m$ and $t \in I_v$, which is a contradiction since $c_M$ is the only point in the intersection $\bigcap_i D(c_i; \mu_M r_i)$. Therefore, $c_M$ is in the convex hull.

Now, define $\hat{c}_i := c_i - c_M$ and let $\theta_{ij}$ denote the angle between vectors $\hat{c}_i$ and $\hat{c}_j$. Since $c_M$ is in the convex hull of $\{c_1,...,c_m\}$, then the vector $0\in \mathbb R^d$ can be written as a convex combination  $\sum_{j=1}^m a_j \hat{c}_j = 0$.
Thus $\langle \sum_{j=1}^m a_j \hat{c}_j, \hat{c}_i\rangle = 0$ for any $i=1,\ldots, m$, and
\[ \sum_{j=1}^m \langle a_j \hat{c}_j, \hat{c}_i \rangle = \sum_{j=1}^m a_j \Vert \hat{c}_j \Vert \, \Vert \hat{c}_i \Vert \cos \theta_{ji} = 0. \]
Taking out the common factor $\Vert \hat{c}_i \Vert$, we have $ \sum_{j=1}^m a_j \Vert \hat{c}_j \Vert\cos \theta_{ji} = 0$.
Now, taking the sum over $i$, we deduce that
\begin{equation} \label{Eq:convex-sum}
\sum_{i=1}^m \sum_{j=1}^m a_j \Vert \hat{c}_j \Vert\cos \theta_{ji} = \sum_{j=1}^m a_j \Vert \hat{c}_j \Vert \sum_{i=1}^m \cos \theta_{ji} = 0.
\end{equation}
Note that $\cos \theta_{ii} = 1$. On the other hand, if were $\cos \theta_{ji} > -\dfrac{1}{m-1}$ for all $1\leq i,j \leq m$, $i \neq j$, then, for each $j$ we should have $\sum_{i=1}^m \cos \theta_{ji} = 1 + \sum_{j=1, j\neq i}^m \cos \theta_{ji} > 0$. But, this contradicts (\ref{Eq:convex-sum}) because $\sum_{j=1}^m a_j = 1$ and $a_j \geq 0$. Therefore, there must exist $i\neq j$, say $i=1$ and $j=2$, such that $\cos \theta_{12} \leq -\dfrac{1}{m-1} \leq -\dfrac{1}{d}$, then
\begin{equation} \label{Eq:inequality-cos}
0 \leq \dfrac{d}{d-1}(1+\cos(\theta_{12})) \leq 1.
\end{equation}
It follows from inequality above and from the AM-GM inequality, that
\begin{equation} \label{MG-MA}
\sqrt{\dfrac{d}{d-1}(1+\cos(\theta_{12}))}\cdot \sqrt{r_1 r_2} \leq \sqrt{r_1 r_2} \leq \dfrac{r_1 + r_2}{2}.
\end{equation}
A straightforward calculation on (\ref{MG-MA}) leads us to the following inequality
\begin{equation*}
(d+1)(r_1 + r_2)^2 \leq 2d(r_1^2 + r_2^2 - 2r_1r_2\cos(\theta_{12}))
\end{equation*}
so,
 \begin{align*}
(d+1)\mu_M^2(r_1 + r_2)^2 	&\leq 2d(\mu_M^2 r_1^2 + \mu_M^2 r_2^2 - 2\mu_M^2 r_1 r_2 \cos(\theta_{12}))\\
							&= 2d(\Vert \hat{c}_1 \Vert^2 + \Vert \hat{c}_2 \Vert^2 - 2 \Vert \hat{c}_1 \Vert \Vert \hat{c}_2 \Vert \cos(\theta_{12}))\\
							&= 2d(\Vert \hat{c}_1 \Vert^2 + \Vert \hat{c}_2 \Vert^2 - 2 \langle \hat{c_1}, \hat{c_2} \rangle)\\
							&= 2d\Vert \hat{c}_1 - \hat{c}_2 \Vert^2
\end{align*} 
which implies, 
\[ \mu_M^2 \leq  \dfrac{2d}{d+1}\cdot \dfrac{\Vert \hat{c}_1 - \hat{c}_2 \Vert^2 }{(r_1 + r_2)^2} = \dfrac{2d}{d+1}\cdot \dfrac{\Vert c_1 - c_2 \Vert^2 }{(r_1 + r_2)^2} \leq \dfrac{2d}{d+1} \]
(the last inequality holds because $M$ is a Vietoris-Rips system) or equivalently $\mu_M \leq \sqrt{\frac{2d}{d+1}}$. 

For a collection with more than $d+1$ disks, the claim of the lemma is a consequence of the Helly's Theorem (see \cite[Problem 29]{Bollobas:2006}) which establishes that for any finite collection, with at least $d+1$ convex subsets of the $d$-dimensional euclidean space $\mathbb{R}^d$, if the intersection of every subcolection with $d+1$ of such sets is nonempty, then the whole collection has a nonempty intersection. This concludes the proof. 
\end{proof}

The upper bound $\sqrt{2d/(d+1)}$ in Lemma \ref{Lemma:VR-generalized} is optimal: it suffices to take a disk system with $d+1$ disks of equal radii and pairwise tangents (cf. \cite[Sec. III.2]{edelsbrunner:2010}).

In the Example \ref{Example:VR-Cech-system} we can see what the Vietoris-Rips Lemma claim for the 2-disk system $M$: $\mu_M = 0.9188 < 0.9812 = \sqrt{4/3}\nu_M$.

To conclude this section, notice that as for a \v{C}ech system $\mu_M\leq 1$, then the Vietoris-Rips Lemma implies the following result.

\begin{corollary}\label{Corollary:Cech-system}
If $M$ is an arbitrary $d$-disk system and $\nu_M$ is its Vietoris-Rips scale, then its \v{C}ech scale satisfies that $\mu_M \in [\nu_M, \sqrt{2d/(d+1)}\, \nu_M]$. In consequence, for every $d$-disk system $M$, the rescaled disk system $M_{\sqrt{2d/(d+1)}\, \nu_M}$ is always a \v{C}ech system.

In particular, if $\sqrt{2d/(d+1)}\, \nu_M \leq 1$ then $M_{\nu_M}$ is a \v{C}ech system. 
\end{corollary}

\section{Filtered generalized simplicial structures for disk systems} \label{Section:filtered-simplicial-structures}

In this section we introduce two simplicial structures associated with a disk system $M$, as well as the filtration induced by rescaling the system $M$. The importance of these notions lies in their relation to the topological analysis through persistent homology of filtered simplicial structures, induced by point clouds with non-homogeneous neighborhoods.

Let $M$ be a disk system. Denote by $\mathcal{VR}(M)$ the family of all Vietoris-Rips subsystems of $M$, this is,
\[ \mathcal{VR}(M) = \{ \sigma \subset M \mid D_i \cap D_j \neq \emptyset, D_i, D_j \in \sigma \}. \]
Analogously, denote by $\mathscr{C}(M)$ the set of all \v{C}ech subsystems,
\[ \mathscr{C}(M) = \{ \sigma \subset M \mid \cap_{D_i\in \sigma} D_i \neq \emptyset\}. \]
 
On the other hand, recall that a simplicial structure on a (finite) set $V$ is defined as a family $\Delta(V) \subset 2^V$ of subsets of $V$ such that: if $\sigma \in \Delta(V)$ and $\tau \subset \sigma$, then $\tau \in \Delta(V)$. Is immediate that for any $\sigma \in \mathcal{VR}(M)$, every disk subsystem $\tau \subset \sigma$ is also in the family $\tau \in \mathcal{VR}(M)$. The same property is valid for the family $\mathscr{C}(M)$. These properties imply that $\mathcal{VR}(M)$ and $\mathscr{C}(M)$ are simplicial complexes. 

We refer to $\mathcal{VR}(M)$ as the generalized Vietoris-Rips complex associated to the disk system $M$, and to $\mathscr{C}(M)$ as the generalized \v{C}ech complex of $M$.

The above construction allows us to study the topology of a data cloud through the persistent homology of the generalized Vietoris-Rips or \v{C}ech complexes, i.e. to perform a topological data analysis. However, to perform such analysis it is necessary to construct a filtered simplicial structure. We will define a filtration through weight functions.

Let $\Delta$ be a simplicial complex and let $\omega : \Delta \to \mathbb{R}$ be a function. We call $\omega$ a \textit{weight function over the simplicial complex $\Delta$} if $\tau, \sigma \in \Delta$ and $\tau \subset \sigma$, implies $\omega(\tau) \leq \omega(\sigma)$.

For example, to the generalized \v{C}ech complex $\mathscr{C}(M)$ of the disk system $M$, the function $\omega : \mathscr{C}(M) \to \mathbb{R}, \sigma \mapsto \mu_\sigma$ which assigns the \v{C}ech scale to any \v{C}ech subsystem $\sigma \subset M$, is a weight function, \textit{called the \v{C}ech-weight function}. The analogous property holds for the Vietoris-Rips complex and the Vietoris-Rips scale (see \cite{MartinTesis:2017} for the construction of the filtered generalized \v{C}ech complex using weighted point clouds).

Moreover, from the definition we have for the \v{C}ech-weight function and to every nonnegative scale $\lambda \geq 0$, that 
\[\omega^{-1}((-\infty, \lambda]) = \mathscr{C}(M_\lambda).\]

We will denote by $\mathscr{C}_{M}(\lambda)$ the family $\mathscr{C}(M_{\lambda})$ for $\lambda \geq 0$, i.e. the family of all \v{C}ech subsystems of the $\lambda$-rescaled disk system $M_{\lambda}$, in order to make the dependence explicit with respect the the parameter. We establish the analogous definition for $\mathcal{VR}_{M}(\lambda)$, for any $\lambda \geq 0$.

It is important to note that there is no restriction on the scale $\lambda \geq 0$, additional to the non-negativity, i.e. we allow greater values of $\lambda$ than 1, in the interest of studying the generalized \v{C}ech complex of rescaled disk systems beyond the original.

For $\lambda' \leq \lambda$ we have the families contention: $\mathcal{VR}_{M}(\lambda') \subset \mathcal{VR}_{M}(\lambda)$ and $\mathscr{C}_{M}(\lambda') \subset \mathscr{C}_{M}(\lambda)$. In general, given a simplicial complex $\Delta$ and a weight function $\omega : \Delta \to \mathbb{R}$, any increasing sequence $\lambda_1 < \cdots < \lambda_s$ of real numbers induces a simplicial filtration: $\Delta_1 \subset \cdots \subset \Delta_s$ for $\Delta_i := \omega^{-1}((-\infty, \lambda_i])$. Thus, for any disk system $M$, the generalized \v{C}ech complex $\mathscr{C}(M)$ has a filtered simplicial complex structure,
\[ \mathscr{C}_M(0) \subset \mathscr{C}(\lambda_1) \subset \cdots \subset \mathscr{C}_M(\lambda_s). \]

Of course, when we vary the scale $\lambda$ on a interval the above filtration contains only a finite number of different sets. Moreover, those sets only change when the \v{C}ech scale of some disk system is reached and therefore it is enough to compute all sets corresponding to \v{C}ech scales to characterize entirely the filtration. The goal of the next sections is the construction of algorithms to numerically estimate the \v{C}ech scale of every \v{C}ech subsystem of $M$.

The filtered generalized \v{C}ech complex can be \textit{approximated} by the Vietoris-Rips structure, in the following sense. The inclusion $\mathscr{C}(M) \subset \mathcal{VR}(M)$ holds clearly, in consequence: $\mathscr{C}_{M}(\lambda) \subset \mathcal{VR}_{M}(\lambda)$ for any scale $\lambda \geq 0$, then by Lemma \ref{Lemma:VR-generalized} any Vietoris-Rips $d$-system $\sigma \in \mathcal{VR}(M)$ rescaled by a factor of $\sqrt{2d/(d+1)}$ is also a \v{C}ech $d$-system: $\sigma_{\sqrt{2d/(d+1)}} \in \mathscr{C}(M)$. Therefore, for any $d$-disk system $M$ the following relation is fulfilled:
\[ \mathcal{VR}_{M}(\lambda') \subset \mathscr{C}_{M}(\lambda) \subset \mathcal{VR}_{M}(\lambda), \]
where $\sqrt{2d/(d+1)} \cdot \lambda' \leq \lambda$.

To any disk system $M$, the simplicial substructure $\mathscr{C}(M)^{(1)}$ given by the 1-skeleton of the generalized \v{C}ech complex of $M$ is a basic combinatorial structure (actually, a graph) that can be easily defined, it just takes the relationship into account if every two vertices are neighbors: the set of vertices is $M$, and there exists an edge $\{D_i, D_j\}$ whenever $D_i \cap D_j \neq \emptyset$. The \v{C}ech-weight function restricted to $\mathscr{C}(M)^{(1)}$ is, in fact: $\omega(\{D_i\}) = 0$ to every vertice, and $\omega(\{D_i, D_j\}) = \Vert c_i - c_j \Vert / (r_i + r_j)$ to any edge.

In Algorithm \ref{Algorithm:Cech-weight-function} we calculate the \v{C}ech-weight function $\omega : \mathscr{C}_M(\lambda)^{(dim)} \to \mathbb{R}$, for the \textit{dim}-skeleton of a $\lambda$-rescaled disk system $M$. To do this, we assume an arbitrary linear order in the disk system $M$, and for every disk $D \in M$ we consider the following set:
\begin{center}
$\lambda$-\texttt{LowerNbrs}$(D) = \{\tilde{D} \in M \mid \tilde{D} < D,\ \omega(\{D, \tilde{D}\}) \leq \lambda\}$.
\end{center}

The following algorithm (based in \cite{Zomorodian2010FastCO}), is a standard expansion algorithm for simplicial complexes, and we are including the \v{C}ech-weight function value of each simplex when it is calculated.
\newpage

\begin{algorithm}\label{Algorithm:Cech-weight-function}
    \SetKwInOut{Input}{Input}
    \SetKwInOut{Output}{Output}

    \Input{A $d$-disk system $M$, a nonnegative parameter $\lambda$ and an integer $dim \geq 2$.}
    \Output{The \v{C}ech-weight function $\omega: \mathscr{C}_M(\lambda)^{(dim)} \to \mathbb{R}$.}
%    Calculate $\mathscr{C}_M(\lambda)^{(1)}$\;
	\For{$n \leftarrow 0${\normalfont\textbf{ to }}$dim-1$}{
    	\For{$\sigma \in \mathscr{C}_M(\lambda)^{(n)}\backslash \mathscr{C}_M(\lambda)^{(n-1)}$}{
    		$\mathrm{LN}_\sigma \leftarrow \bigcap_{D \in \sigma}$ $\lambda$-\texttt{LowerNbrs}$(D)$\;
    		\For{$D \in \mathrm{LN}_\sigma$}{
	    		$N \leftarrow \{D\}\cup\sigma$\;
    			\eIf{$n \leq d-1$}{
    				Calculate $\omega(N) \leftarrow \mu_N$\;
    			}{	$\omega(N) \leftarrow \max\{\omega(\tau) \mid \tau \varsubsetneq N \}$\;
    			}
    		\If{$\omega(N) \leq \lambda$}{
    			Update $\mathscr{C}_M(\lambda)^{(dim)} \leftarrow \mathscr{C}_M(\lambda)^{(dim)} \cup \{N\}$\;
    		}
    		}
    	}
    }
	\Return($\mathscr{C}_M(\lambda)^{(dim)}, \omega: \mathscr{C}_M(\lambda)^{(dim)} \to \mathbb{R}$)
    \caption{\v{C}ech-weight function of a $d$-disk system.}
\end{algorithm}

We conclude the section with an application of Algorithm \ref{Algorithm:Cech-weight-function} to a 2-disk system.

\begin{example}\normalfont
Let $M$ be the following 2-disk system,
\begin{align*}
M = \{ & D_1((2.99,0.56); 1.5), D_2((0.99,0.11); 1.0), D_3((1.69,1.30); 0.6),\\ 				& D_4((1.07,1.93); 0.4), D_5((1.96,2.64); 0.8) \}.
\end{align*}

The output of Algorithm \ref{Algorithm:Cech-weight-function} applied to $M$, with $d=2$, $\lambda=1$ and $dim=2$ gives the \v{C}ech scales indicated next to every edge and in the triangle, in Figure \ref{Figure:Filtered-Cech-system}. The \v{C}ech scale of the 2-disk system $\{D_1, D_2, D_3\}$ was calculated with the \texttt{Cech.scale} script from Algorithm \ref{Algorithm:Cech-scale-triplets}.
\begin{figure}[H]
\centering
\includegraphics[width=0.45\textwidth]{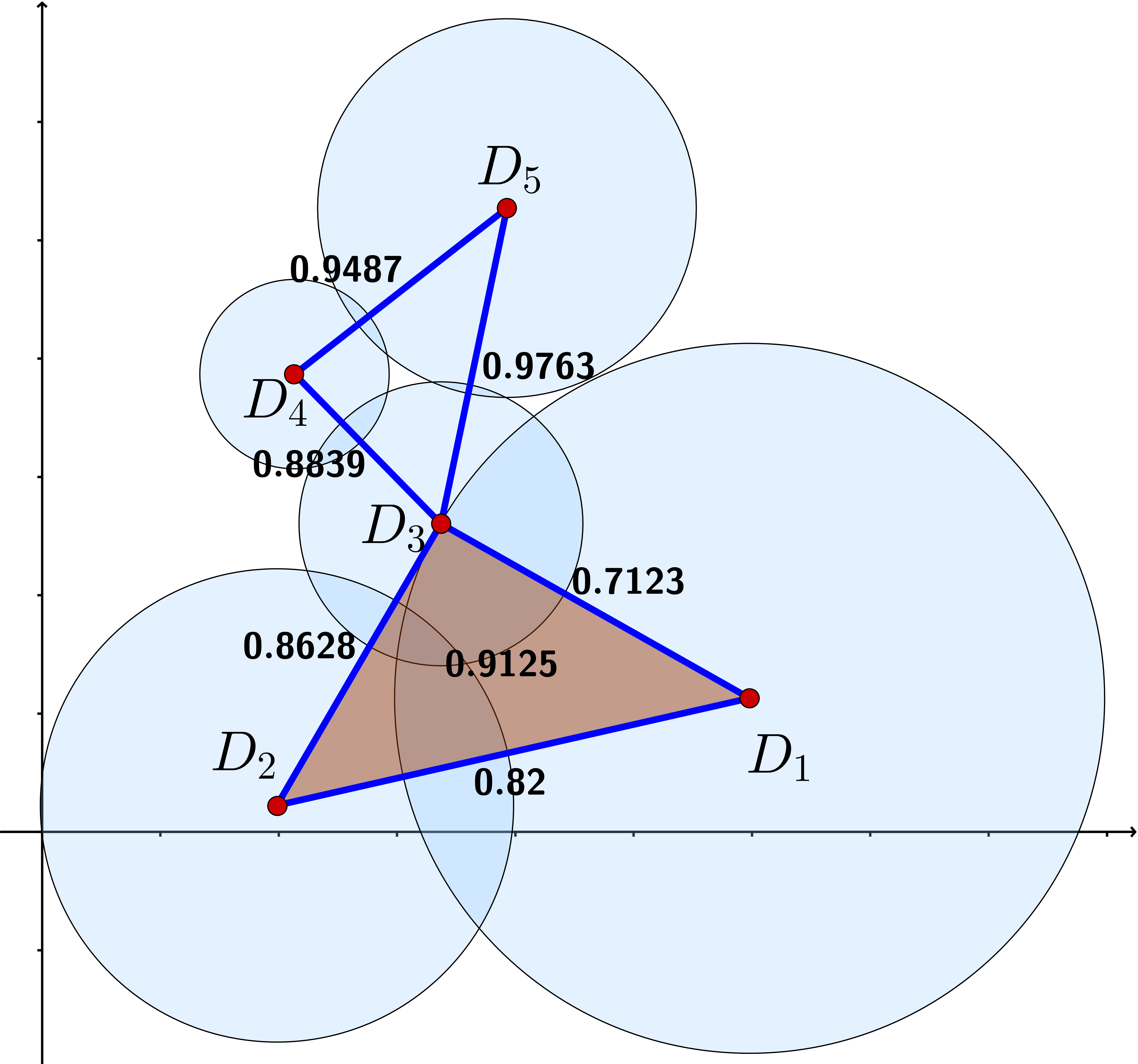}
\caption{The \v{C}ech-weight function of the 2-disk system $M$.}
\label{Figure:Filtered-Cech-system}
\end{figure}
\end{example}

\section{Intersection properties of disk systems} \label{Section:intersection-properties}
In this section we focus on studying disk systems in the plane, i.e. 2-\textit{disk systems}. 
As we have seen in the last section, the study of the \v{C}ech scale is a key aspect to the construction and study of filtered generalized \v{C}ech complex. In this section, we establish several intersection properties of 2-disk systems, which will lead us to be able to calculate the \v{C}ech scale.

Let $\partial D_i(c_i;r_i):= \{x \in \mathbb{R}^2 \mid \Vert x - c_i \Vert = r_i\}$ be the boundary of the closed 2-dimensional disk $D_i(c_i,r_i) \subset \mathbb{R}^2$.

Let $D_i$ and $D_j$ be two closed disks in the plane, such that $D_i \cap D_j \neq \emptyset$. We define $D_i \sqcap D_j$ to be the unitary set $\{d_{ij}\}$ constructed as follows:
\begin{enumerate}
\item If $\partial D_i \cap \partial D_j \neq \emptyset$, then $d_{ij} \in \partial D_i \cap \partial D_j$ is the only one point with the property $\langle d_{ij} - c_i, \mathbf{n}_{ij} \rangle \geq 0$, where $\mathbf{n}_{ij} = (-b, a)$ is the normal vector to $c_j - c_i = (a,b)$,
\item If $\partial D_i \cap \partial D_j = \emptyset$,  we define $d_{ij}$ as the unique intersection point in $\partial D_i(c_i;\lambda r_i) \cap \partial D_j(c_j;\lambda r_j)$, for $\lambda$ given as the minimal scale such that $D_i(c_i;\lambda r_i) \subset D_j(c_j;\lambda r_j)$ or $D_j(c_j;\lambda r_j) \subset D_i(c_i;\lambda r_i)$, i.e. $\lambda = \Vert c_i -c_j \Vert /|r_i - r_j|$.
\end{enumerate}

Clearly, if $\partial D_i \cap \partial D_j = \emptyset$, then $d_{ij} = d_{ji}$. In particular, when $D_i$ and $D_j$ are concentric, then $D_i \sqcap D_j = \{c_i\} = \{c_j\}$. On the other hand, if the closed disks $D_i$ and $D_j$ are internally or externally tangent, then $d_{ij} = d_{ji}$. 
We can think about $d_{ij}$, when $\partial D_i \cap \partial D_j$ is not empty, as the intersection point of the boundaries at the left of the vector from $c_i$ to $c_j$. Figure \ref{Figure:Intersection-points} shows the above construction.

\begin{figure}[hbt]
\centering
\includegraphics[width=0.7\textwidth]{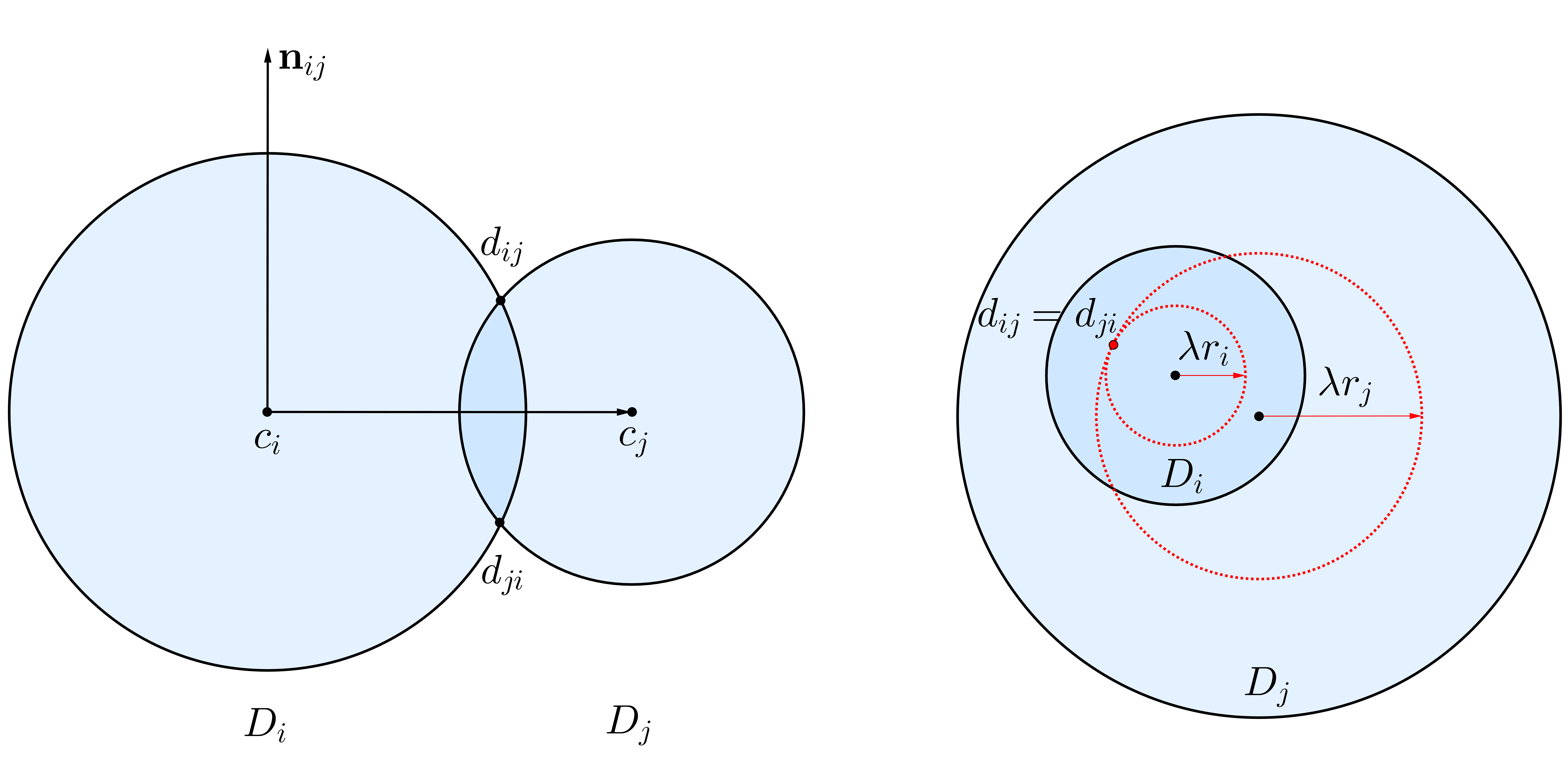}
\caption{Intersection point $d_{ij}$.}
\label{Figure:Intersection-points}
\end{figure}

We will denote by $d_{ij}(\lambda)$, instead of simply $d_{ij}$, for the intersection point of the $\lambda$-rescaled disks $D_i(c_i; \lambda r_i)$ and $D_j(c_j; \lambda r_j)$.

In order to study \v{C}ech systems, we give the following characterization, according the intersection points $d_{ij}$.

\begin{lemma}\label{Lemma:Cech-generalized}
Let $M = \{D_1, D_2, \ldots, D_m\}$ be a 2-disk system. Then $M$ is a \v{C}ech system if and only if, there exist $D_i, D_j \in M$ such that $d_{ij} \in D_i \sqcap D_j$ satisfies $d_{ij} \in D_k$ for all $1 \leq k \leq m$.
\end{lemma}

\begin{proof}
Suppose $M$ is a \v{C}ech system. Define $A := \bigcap_{1\leq i \leq m} D_i \neq \emptyset$. Then $A$ has only one of the following geometries:
\begin{enumerate}
\item[(\textit{i})] $A=\{c_M\}$,
\item[(\textit{ii})] $A$ is a region bounded by more than one circumference arc,
\item[(\textit{iii})] $A=D_{i_0}$ for some $i_0 \in \{1,\ldots, m\}$.
\end{enumerate}

In the first case, necessarily $c_M$ belongs to the boundary of two or more disks. Let $D_i$ and $D_j$ be two disks in $M$ such that $c_M \in \partial D_i \cap \partial D_j$, it follows that $c_M = d_{ij}$ or $c_M = d_{ji}$, in both cases the lemma holds.

For the second case, if $a \in \partial A \subset A$ belongs to the boundary and is in the intersection of two arcs, say $\partial D_i$ and $\partial D_j$, then $a = d_{ij}$ or $a=d_{ji}$, and it satisfy $a \in D_k$ for every $1 \leq k \leq m$. 

For the last case, if $A=D_{i_0}$ for some $i_0$, then for each $j\neq i_0$ we have $d_{i_0 j} \in D_{i_0} = A$ and all of these points belongs to $D_k$ for all $D_k \in M$.

Therefore, in any case there exists such point $d_{ij}$. The converse is clear by definition of a \v{C}ech system.
\end{proof}

This criterion was presented in \cite[Sec. III]{CechGenelized} for a 2-disk system. 

Next, we define the map $\rho$, which in a certain sense quantifies the intersection of a 2-disk system. This map will allow us to discern the minimal scale in which a 2-disk system has the nonempty intersection property.

\begin{definition}\label{Definition:Function-rho}
Let $M=\{D_1,D_2,\ldots , D_m \}$ be a Vietoris-Rips system in the plane, with $m\geq 3$. We define 
\[\rho(M):=\max_{1\leq i,j \leq m} \left\{ \min_{k\neq i,j} \{r_k - \Vert d_{ij} - c_k \Vert \} \right\}. \]
If $\nu_M$ is the Vietoris-Rips scale of $M$, then we define the map $\rho_M : [\nu_M, \infty) \to \mathbb{R}$, $\lambda \mapsto \rho_M(\lambda)=\rho(M_\lambda)$.
\end{definition}

Given three disks $D_i(c_i;r_i)$, $D_j(c_j;r_j)$ and $D_k(c_k;r_k)$ in the 2-disk system $M=\{D_1,\ldots, D_m\}$, with Vietoris-Rips scale $\nu_M$, denote by $\Lambda_{i,j}^k : [\nu_M, \infty) \to \mathbb{R}$ the map $\lambda \mapsto \lambda r_k - \Vert d_{ij}(\lambda) - c_k \Vert$, where $d_{ij}(\lambda)$ is the element in $D_i(c_i; \lambda r_i) \sqcap D_j(c_j; \lambda r_j)$. In other words, $\Lambda_{i,j}^k(\lambda)$ is the \textit{signed distance} from the point $d_{ij}(\lambda)$ to the set $D_k(c_k; \lambda r_k)$.

If $r_i \neq r_j$, then for each $k \neq i,j$ the map $\lambda \mapsto \lambda r_k - \Vert d_{ij}(\lambda) - c_k \Vert$ is defined and continuous in the closed interval $\left[\Vert c_i - c_j \Vert / (r_i + r_j), \Vert c_i - c_j \Vert / |r_i - r_j| \right]$, since it is the signed distance from  an intersection point of two continuously deforming curves (therefore its position vary continuously as long as the intersection exists) to the continuously deforming set $D_k(c_k; \lambda r_k)$ with respect to $\lambda$. Also, the map $\lambda \mapsto \lambda r_k - \Vert d_{ij}(\lambda) - c_k \Vert$ vary linearly in the range $[\Vert c_i - c_j \Vert / |r_i - r_j|, \infty)$ because for $\lambda \geq \Vert c_i - c_j \Vert / |r_i - r_j|$, the term $\Vert d_{ij}(\lambda) - c_k \Vert$ remains constant. The left picture in Figure \ref{Figure:geometric-place} shows in bold red color the geometric place of $\{d_{ij}(\lambda), d_{ji}(\lambda)\}$ which vary continuously respect to the parameter $\lambda$ and also the distance from it to the fix point $c_k$.

\begin{figure}[hbt]
\centering
\includegraphics[width=0.7\textwidth]{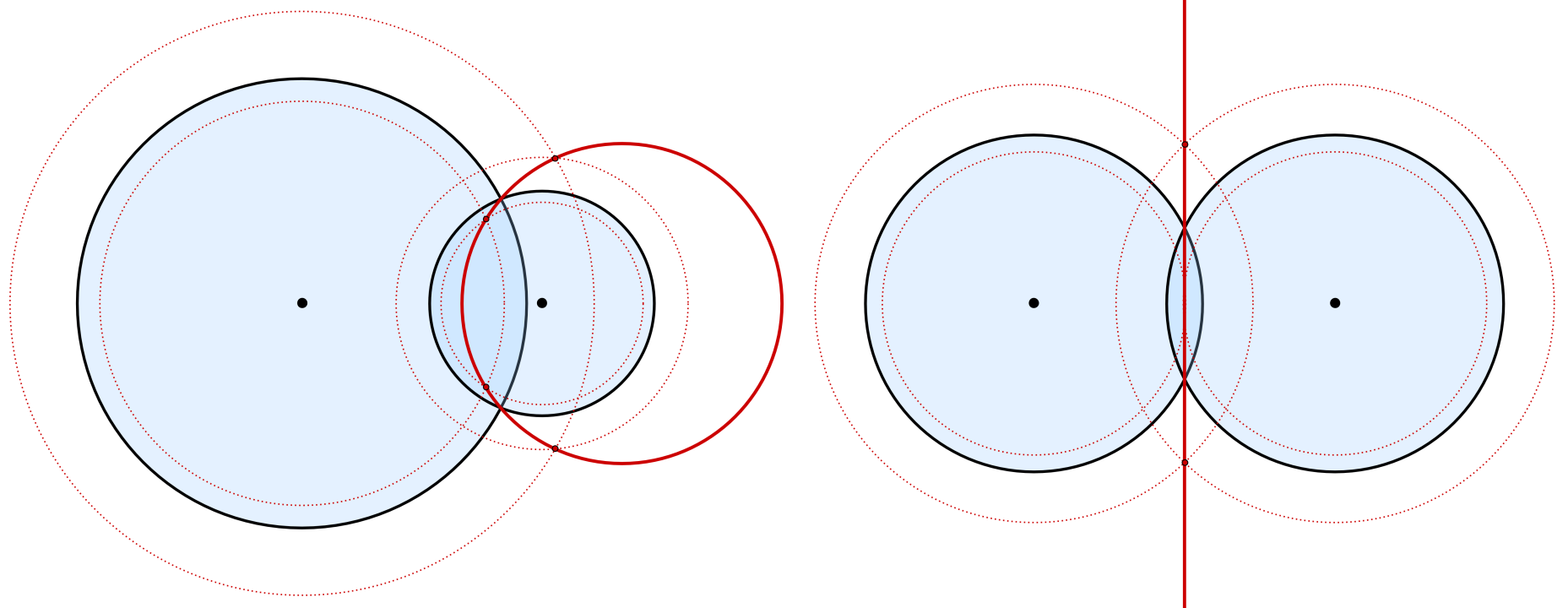}
\caption{Geometric place of $d_{ij}(\lambda)$.}
\label{Figure:geometric-place}
\end{figure}

\noindent On the other hand, for $r_i = r_j$ the points $\{d_{ij}(\lambda), d_{ji}(\lambda)\}$ vary continuously respect to $\lambda$ on the line showed in the right picture of Figure \ref{Figure:geometric-place}. Therefore, $\lambda r_k - \Vert d_{ij}(\lambda) - c_k \Vert$ also depend continuously of $\lambda$.

From the above argument, each map $\Lambda_{i,j}^k$ is continuous in the interval $[\nu_M, \infty)$, and by the continuity of the $\min$-$\max$ functions and that,
\[\rho_M(\lambda) = \max_{1\leq i,j \leq m} \left\{ \min_{k\neq i,j} \Lambda_{i,j}^k(\lambda) \right\}, \]
it follows that $\rho_M$ is also a continuous map in the interval $[\nu_M, \infty)$. However, the map $\rho_M$ is not differentiable, in general, to every point in such interval.

The map $\rho_M$ plays a key role in the rest of this work. We present the next characterization of \v{C}ech systems in terms of $\rho_M$.

\begin{lemma} \label{Lemma:Cech-system-nonnegative-scale}
Let $M$ be a 2-disk system. Then $M_\lambda$ is a \v{C}ech system if and only if, $\rho_M(\lambda) \geq 0$. In particular $\rho_M(\sqrt{4/3}\nu_M) \geq 0$.
\end{lemma}

\begin{proof}
By Lemma \ref{Lemma:Cech-generalized}, $M_\lambda$ is a \v{C}ech system in the plane if and only if, there exists $d_{ij}(\lambda)$ such that  $d_{ij}(\lambda) \in D_k(c_k; \lambda \, r_k)$ for every $k\neq i,j$, i.e. $\Lambda_{i,j}^k(\lambda) \geq 0$ for every $k\neq i,j$, which is equivalent to $\rho_M (\lambda) \geq 0$. 

On the other hand, from Corollary \ref{Corollary:Cech-system} and taking $d=2$, the rescaled system $M_{\sqrt{4/3}\nu_M}$ is a \v{C}ech system, then by the first assertion, $\rho_M(\sqrt{4/3}\nu_M) \geq 0$.
\end{proof}

\section{The \texttt{Cech.scale} algorithm} \label{Section:Cech.scale-algorithm}

Our main algorithm (Algorithm \ref{Algorithm:Cech-scale-triplets}) computes the \v{C}ech scale of a given 2-disk system $M$. The key aspect on which this algorithm is based, is precisely the function $\rho_M$. Before we describe the algorithm, we need to analyze additional properties of $\rho_M$. 

It follows immediately, from Lemma \ref{Lemma:Cech-system-nonnegative-scale}, that $\rho_M(\lambda) \geq 0$ for every $\lambda \geq \mu_M$. 
Also, if at the Vietoris-Rips scale it holds that $\rho_M(\nu_M) \geq 0$, then $\mu_M = \nu_M$ by the minimality of the \v{C}ech scale. We conclude that in this case (this is, $\rho_M(\nu_M) \geq 0$), the \v{C}ech scale is easily computable.

On the other hand, if $\rho_M (\nu_M) < 0$ then the \v{C}ech scale satisfies $\mu_M \in (\nu_M, \sqrt{4/3}\nu_M]$ and moreover $\rho_M(\mu_M) = 0$. This is a consequence of the continuity of $\rho_M$, and the fact that $\rho_M(\lambda) < 0$ for every $\nu_M \leq \lambda < \mu_M$ and $\rho_M(\lambda) \geq 0$ for $\mu_M \leq \lambda$. Thus, to find the \v{C}ech scale of a 2-disk system $M$ for which $\rho_M(\nu_M) < 0$, we need to solve the equation $\rho_M(\lambda) = 0$.

We propose a numerical approach, in order to solve the equation $\rho_M(\lambda) = 0$, to calculate the \v{C}ech scale under the hypothesis $\rho_M(\nu_M) < 0$,  since in this case, we actually know that $\mu_M \in (\nu_M, \sqrt{4/3}\nu_M]$ (see Section \ref{Section:Rips.Cech.systems}) as consequence of the generalized Vietoris-Rips Lemma. We have chosen the \textit{bisection} method for this purpose. We will denote the implementation of bisection method for the map $\rho_M$ through the interval $[a,b]$, by \texttt{bisection}$(\rho_M, a, b)$. The output of \texttt{bisection}$(\rho_M, a, b)$ is a real number $\lambda \in [a,b]$ such that $\rho_M(\lambda) =0$. For the numerical method we are working with a precision of $10^{-12}$.

It is important to mention that the numerical method \textit{regula falsi} was also used instead of the numerical method of bisection, in order to calculating the \v{C}ech scale. However, in our context the efficiency of the program using the \textit{regula falsi} numerical method is not better than if the numerical method of bisection is used.

The Algorithm \ref{Algorithm:Cech-scale} (below) has as input a 2-disk system $M$, and produces as output the \v{C}ech scale $\mu_M$ as well as the intersection point $\{c_M\} =  \bigcap_{D_i \in M} D_i(c_i; \mu_M\, r_i)$. This algorithm takes a naive approach to calculate the \v{C}ech scale, and is established to completeness and to be a reference for the principal algorithm (Algorithm \ref{Algorithm:Cech-scale-triplets}).
%\vspace*{1mm}
\newpage

\begin{algorithm}[hbt] \label{Algorithm:Cech-scale}
    \SetKwInOut{Input}{Input}
    \SetKwInOut{Output}{Output}

    \Input{A 2-disk system $M$.}
    \Output{The \v{C}ech scale $\mu_M$ and the intersection point $c_M$.}
%    Calculate $\nu_M$\;
%    $\mu^* \leftarrow \nu_M$\;
    Calculate $\mu^* \leftarrow \nu_M$\;
    \eIf{$\rho_M(\mu^*) \geq 0$}{ \label{Algorithm-line:rho-nonnegative}
	    \Return($\mu_M = \mu^*, c_M$) \label{Algorithm-line:return-vietoris-rips}
    }{
	    Update $\mu^* \leftarrow \sqrt{4/3}\, \nu_M$\;
    }
	Update $\mu^* \leftarrow$\texttt{bisection}$(\rho_M, \nu_M, \mu^*)$\; \label{Algorithm-line:bisection}
	Calculate $\sqcap M_{\mu^*} \leftarrow \{ d_{ij}(\mu^*) \in D_i \sqcap D_j \mid D_i, D_j \in M_{\mu^*} \} \bigcap \left( \bigcap_{D_i \in M} D_i(c_i; \mu^* \, r_i) \right)$\; \label{Algorithm-line:intersection.point}
	\If{$|\sqcap M_{\mu^*}| > 1$}{\label{Algorithm-line:cardinality}
	Find $\mu' \in (\nu_M, \mu^*)$ such that $\rho_M(\mu') > 0$\; \label{Algorithm-line:start-IF}
	Update $\mu^* \leftarrow \mu'$\;
	go to (\ref{Algorithm-line:bisection})\;
	}\label{Algorithm-line:end-IF}
	\Return($\mu_M = \mu^*, c_M$) \label{Algorithm-line:Returns-Cech-scale}
    \caption{The \v{C}ech scale calculation for a 2-disk system.}
\end{algorithm}
%\bigskip

The following lemma claims that the Algorithm \ref{Algorithm:Cech-scale} is consistent.

\begin{lemma} \label{Lemma:Cech.scale}
For any 2-disk system $M$, the Algorithm \ref{Algorithm:Cech-scale} has as output the \v{C}ech scale $\mu_M$ of $M$, and the unique intersection point $\{c_M\} = \bigcap_{D_i \in M} D_i(c_i; \mu_M \, r_i)$.
\end{lemma}

\begin{proof}
In the case $\rho_M(\nu_M) \geq 0$, it is clear that the algorithm has generated asseverated data (steps (\ref{Algorithm-line:rho-nonnegative})-(\ref{Algorithm-line:return-vietoris-rips})). In otherwise, for the case $\rho_M(\nu_M) < 0$, we assign $\mu^* := \sqrt{4/3}\, \nu_M$. 

Then, $\rho_M(\nu_M)\cdot \rho_M(\mu^*) \leq 0$ and lets call again $\mu^*$ the output root in step (\ref{Algorithm-line:bisection}). To check if $\mu^*$ is the \v{C}ech scale we are looking for, we calculate in step (\ref{Algorithm-line:intersection.point}) the set of pairwise intersection points of the $\mu^*$-rescaled system, contents in $\bigcap_{D_i \in M} D_i(c_i; \mu^* \, r_i)$. 

If the set $\sqcap M_{\mu^*}$ is unitary, then necessarily $\bigcap_{D_i \in M} D_i(c_i; \mu^* \, r_i)$ is unitary, due the geometry of its boundary (see proof of Lemma \ref{Lemma:Cech-generalized}). In such case (negative validation of step (\ref{Algorithm-line:cardinality})) the steps (\ref{Algorithm-line:start-IF})-(\ref{Algorithm-line:end-IF}) are omitted and, from Lemma \ref{Lemma:unique-point}, the algorithm returns the \v{C}ech scale as well as the intersection point $\sqcap M_{\mu^*}$ at step (\ref{Algorithm-line:Returns-Cech-scale}); in otherwise (positive validation of step (\ref{Algorithm-line:cardinality})), the root $\mu^*$ is not the \v{C}ech scale (see Figure \ref{Figure:rho-plot}), and then we should find another scale $\mu' \in (\nu_M, \mu^*)$ such that $\rho_M(\mu') > 0$, and repeat from step (\ref{Algorithm-line:bisection}). It is possible, for some configurations of the 2-disk system, that the map $\rho_M$ has a behavior as in Figure \ref{Figure:rho-plot}.

\begin{figure}[hbt]
\centering
\includegraphics[width=0.5\textwidth]{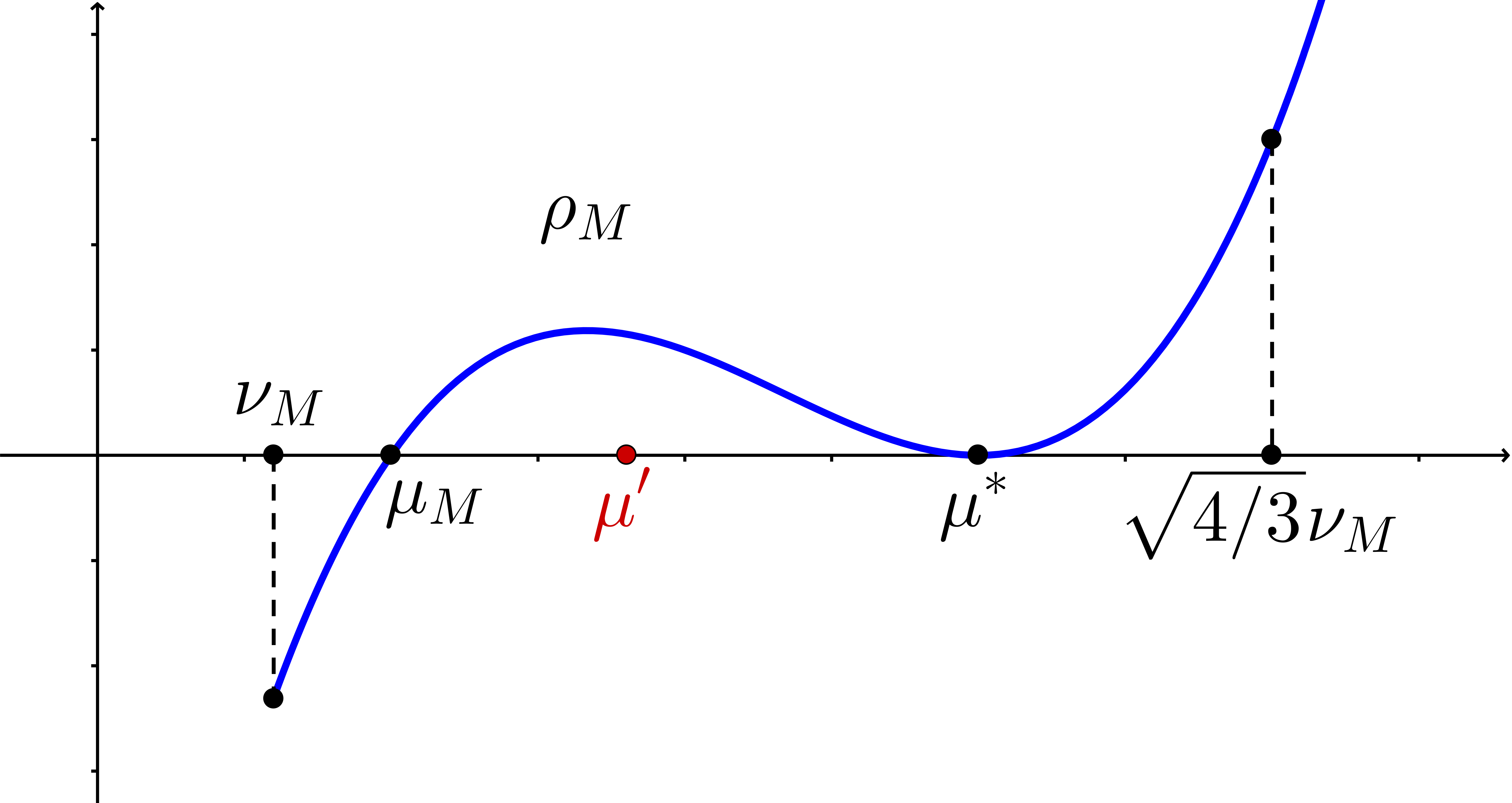}
\caption{Plot of the map $\rho:[\nu_M, \infty) \to \mathbb{R}$.}
\label{Figure:rho-plot}
\end{figure}

The last iterative part is a finite process because $\rho_M$ is algebraic over $\mathbb{Q}$, then eventually the set $\sqcap M_{\mu^*}$ will be unitary and the \v{C}ech scale will be calculated.
\end{proof}

The step (\ref{Algorithm-line:intersection.point}) in Algorithm \ref{Algorithm:Cech-scale} is necessary, as show the following example, in which the map $\rho_M$ has another root along side the \v{C}ech scale in the interval $[\nu_M, \sqrt{4/3}\nu_M]$.

\begin{example} \label{Example:Figure:colinear-disks-system}\normalfont
Let $M = \{D_1, D_2, D_3\}$ be the 2-disk system in Figure \ref{Figure:colinear-disks-system}. A direct calculation, yields that $\nu_M = \mu_M = 0.8947$. On the other hand, we also have that $\rho_M(\nu_M) = \rho_M(\mu_M) = \rho_M(\lambda) = 0$ for $\lambda = 1$. Therefore, the map $\rho_M$ has more than one root on the interval $[\nu_M, \sqrt{4/3}\nu_M] = [0.8947, 1.0331]$.

\begin{figure}[hbt]
\centering
\includegraphics[width=0.6\textwidth]{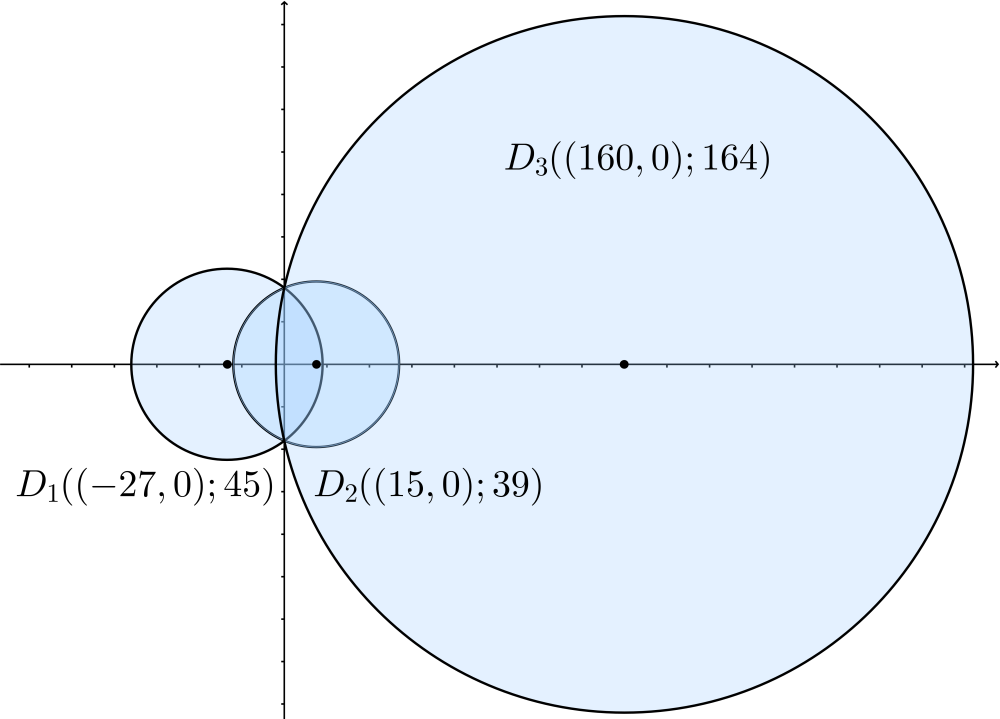}
\caption{The 2-disk system $M$.}
\label{Figure:colinear-disks-system}
\end{figure}
\end{example}

In Example \ref{Example:Figure:colinear-disks-system} the Vietoris-Rips scale $\nu_M$, of the 2-disk system $M$, agrees with the \v{C}ech scale $\mu_M$; however is possible to construct more sophisticated (and symmetric) disk system $M$ such that $\nu_M < \mu_M$ and $\rho_M(\nu_M) < 0 < \rho_M(\mu_M)$, for which there exists $\mu' \in (\nu, \sqrt{4/3}\nu_M)$ with $\rho_M(\mu')=0$.

On the other hand, if the 2-disk system $M$ consists of just three disks and $\rho_M(\nu_M) < 0$, then its \v{C}ech scale can be computed with only one application of the numerical method, as we asseverate in the following lemma.
%\newpage

\begin{lemma} \label{Lemma:triplets}
Let $M=\{D_1, D_2, D_3\}$ be a 2-disk system such that $\rho_M(\nu_M) < 0$. Then, there exists a unique root of the map $\rho_M$ in $[\nu_M,  \sqrt{4/3}\nu_M]$. Thus, $\mu_M$ will be the output of {\normalfont \texttt{bisection}}$(\rho_M, \nu_M, \sqrt{4/3}\nu_M)$.
\end{lemma}

\begin{proof}
It is straightforward to verify that $\rho_M(\nu_M) \geq 0$ for any configuration with $c_1$, $c_2$ and $c_3$ collinear. Thus, it follows $\{c_1, c_2, c_3\}$ is in general position.

Let $\mu_M$ be the \v{C}ech scale of the 2-disk system $M$ and $c_M$ the intersection point. Define $A(\lambda) := \cap_{i=1}^{3} D_i(c_i; \lambda r_i)$. Note that $A(\mu_M) = \{c_M\}$.

We claim that there exist at least two distinct points $p$ and $q$, in the set 
\[\{ d_{ij}(\lambda) \in D_i(c_i; \lambda r_i) \sqcap D_j(c_j; \lambda r_j) \mid D_i, D_j \in M_{\lambda}\} \subset \partial A(\lambda).\] This is evident if $\partial A(\lambda)$ is given by two or more circumference arcs. On the other hand, if $\partial A(\lambda) = \partial D_i(c_i; \lambda r_i)$ for some $1 \leq i \leq 3$, then $d_{ij}(\lambda), d_{ik}(\lambda) \in D_i(c_i; \lambda r_i)$ for $j \neq i$ and $k \neq i$. Moreover, $d_{ij}(\lambda) \neq d_{ik}(\lambda)$ since $\{c_1, c_2, c_3\}$ is not a collinear set.

If both points $p$ and $q$ belongs to each boundary of the three disks, $p,q \in \cap_{i=1}^{3} \partial D_i(c_i; \lambda r_i)$, then $\{c_1, c_2, c_3\}$ would be also a collinear set. Without loss of generality, we suppose that $p=d_{23}(\lambda) \not\in \partial D_1(c_1; \lambda r_1)$. Then, $\rho_M(\lambda) \geq \lambda r_1 - \Vert d_{23}(\lambda) - c_1 \Vert > 0$, and the lemma follows. Of course, the choose of the indexes depend of the value of $\lambda$, but the above arguments show that always there exist such combination which guarantee that $\rho_M(\lambda)$ is positive for $\lambda > \mu_M$.
\end{proof}

The following algorithm takes advantage of the unicity property for the root of $\rho_M$, in a 2-disk system with three disks. Essentially, the algorithm consist in iterating the Algorithm \ref{Algorithm:Cech-scale} systematically over every triplet of disks from $M$.
%\vspace*{2mm}

\begin{algorithm}[hbt] \label{Algorithm:Cech-scale-triplets}
    \SetKwInOut{Input}{Input}
    \SetKwInOut{Output}{Output}

    \Input{A 2-disk system $M$.}
    \Output{Cech scale $\mu_M$ and intersection point $c_M$.}
    Calculate $\nu_M$\; \label{Algorithm-line:Triplets-Rips-scale}
    $\mu^* \leftarrow \nu_M$\;
    \If{$\rho_M(\mu^*) \geq 0$}{
    	\Return($\mu_M=\mu^*; c_M$)
    }\label{Algorithm-line:Triplets-return-VR}
    \For{$1 \leq i < j < k \leq |M|$}{ \label{Algorithm-line:Triplets-start-for}
    	$N \leftarrow \{D_i, D_j, D_k\}$\;
    	Calculate $\nu_N$\;
    	\If{$\sqrt{4/3}\, \nu_N \geq \mu^*$}{ \label{Algorithm-line:Triplets-condition}
	    	\eIf{$\rho_N(\nu_N) \geq 0$}{
	    		$\mu_N \leftarrow \nu_N$\;
	    	}{	$\mu_N \leftarrow$ \texttt{bisection($\rho_N, \nu_N, \sqrt{4/3}\, \nu_N$)}}
	   		\If{$\mu_N > \mu^*$}{\label{Algorithm-line:Triplets-start-if}
	   			Update $\mu^* \leftarrow \mu_N$\;
	   		}\label{Algorithm-line:Triplets-end-if}
    	}
    }\label{Algorithm-line:Triplets-end-for}
    \Return($\mu_M = \mu^*; c_M$)
    \caption{\texttt{Cech.scale}.}
\end{algorithm}

\begin{theorem} \label{Theorem:Cech.scale}
For any 2-disk system $M$, the Algorithm \ref{Algorithm:Cech-scale-triplets} has as output the \v{C}ech scale $\mu_M$ of $M$, and the unique intersection point $\{c_M\} = \bigcap_{D_i \in M} D_i(c_i; \mu_M \, r_i)$.
\end{theorem}

\begin{proof}
If $\rho_M(\nu_M) \geq 0$, the algorithm returns the right data: steps (\ref{Algorithm-line:Triplets-Rips-scale})-(\ref{Algorithm-line:Triplets-return-VR}).

On the other hand, by Helly's Theorem (cf. \cite{Bollobas:2006}) the 2-disk system $M = \{D_1, \ldots, D_m\}$, as a finite family of convex sets in the plane, has a nonempty intersection $\bigcap_{D_i \in M} D_i$ if, and only if, $D_i \cap D_j \cap D_k \neq \emptyset$ for every triplet $1 \leq i < j < k \leq m$. Let $\mu$ be the maximal \v{C}ech scale over every triplet in the disk system $M$, i.e. 
\[ \mu = \max \{ \mu_{N} \mid N = \{D_i, D_j, D_k\} \subset M \} .\]
It follows that every $\mu$-rescaled triplet has a nonempty intersection. Hence, the $\mu$-rescaled 2-disk system $M_\mu$ also has the nonempty intersection property. Moreover, $\bigcap_{D_i \in M} D_i(c_i; \mu r_i) \subset \bigcap_{D_i \in N \subset M} D_i(c_i; \mu r_i)$ for every triplet $N \subset M$. Therefore, $\mu$ is actually the \v{C}ech scale of the 2-disk system $M$, this is, $\mu_M = \mu$.

In steps (\ref{Algorithm-line:Triplets-start-for})-(\ref{Algorithm-line:Triplets-end-for}) the algorithm search the scale $\mu$ systematically, over every triplet $\{D_i, D_j, D_k\} \subset M$, updating the maximal scale found if necessary in steps (\ref{Algorithm-line:Triplets-start-if})-(\ref{Algorithm-line:Triplets-end-if}). By Lemma \ref{Lemma:triplets}, every \v{C}ech scale calculation over any triplet, requires just one application of the bisection method. This implies the correctness of the algorithm.

Additionally, the condition in step (\ref{Algorithm-line:Triplets-condition}) avoids calculating unnecessary \v{C}ech scales of triplets $N=\{D_i, D_j, D_k\}$. In effect, if $\lambda^*$ is the maximal \v{C}ech scale found until the verification of the triplet $N$, and the condition in step (\ref{Algorithm-line:Triplets-condition}) does not satisfy, i.e.
\[ \mu_N \leq \sqrt{4/3}\nu_M < \lambda^* ,\]
then, whatever is the \v{C}ech scale of $N$, it would be not greater than $\lambda^*$.
\end{proof}

The computational evidence to support the Algorithm \ref{Algorithm:Cech-scale-triplets} is more efficient than Algorithm \ref{Algorithm:Cech-scale}, is given in Figure \ref{Graph:Cech.scale-vs-triplets}. The graphic shows the average time (in seconds) to computation of both algorithms, with respect to the number of disks in a randomly generated 2-disk system (see Remark \ref{Remark:computer}).

\begin{figure}[hbt]
\centering
\includegraphics[width=0.8\textwidth]{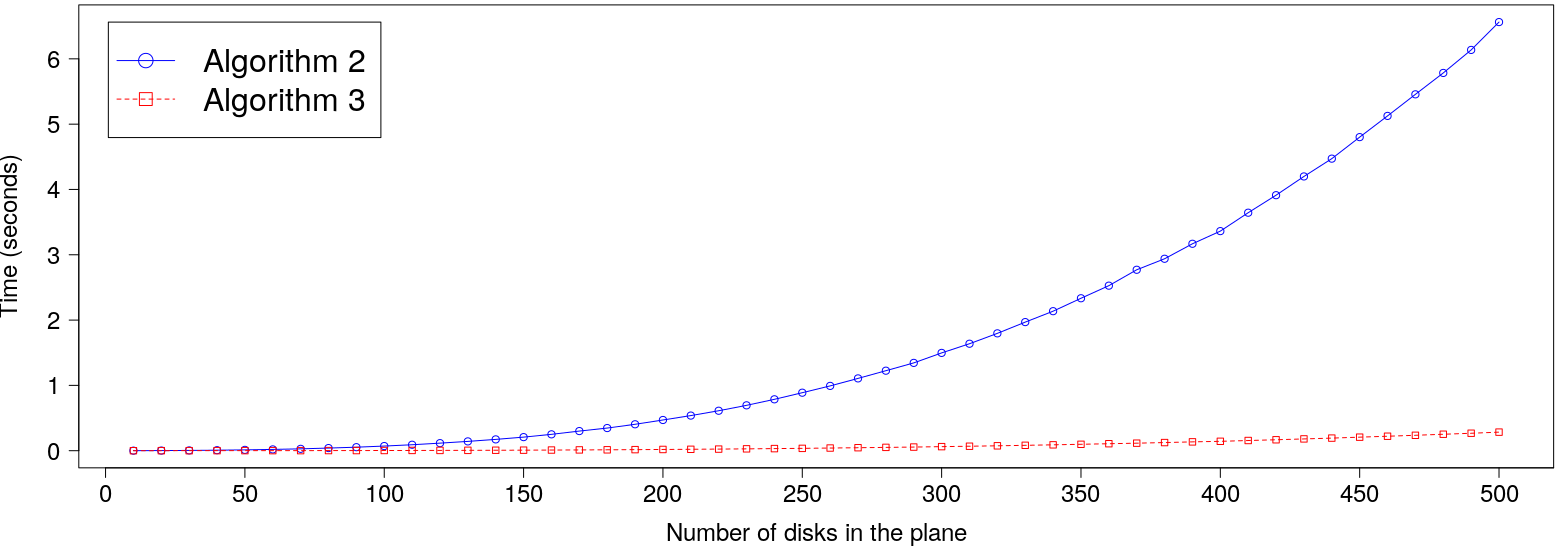}
\caption{Average time of the script \texttt{Cech.scale}.}
\label{Graph:Cech.scale-vs-triplets}
\end{figure}

\section{Applications} \label{Section:applications}

\subsection{The miniball problem}
The miniball problem or smallest-circle problem in the euclidean space is a classical problem, proposed by James J. Sylvester in 1857. 

Given a finite point cloud $N \subset \mathbb{R}^d$ the miniball problem consists in finding the center $c\in \mathbb{R}^d$ and minimum radius $r \in \mathbb{R}^+$ of a $d$-disk $D = D(c;r) \subset \mathbb{R}^d$ such that $N \subset D$.

There exist many different approaches to solve this problem, and a variety of algorithms to reach the miniball data (e.g. \cite{Fischer:2003, Welzl:1991}). In fact, the \v{C}ech scale has a close relation with the miniball problem, as we establish in the next lemma. 

\begin{lemma}
Let $N$ be a finite point cloud in $\mathbb{R}^d$, and let $N_1$ be the associated $d$-disk system defined by
\[ N_1:= \{D_i(c_i; 1) \subset \mathbb{R}^d \mid c_i \in N\}. \]
Then, the \v{C}ech scale $\mu_{N_1}$ is the radius of the minimal enclosing ball of $N$, and the intersection point $\{c_{N_1}\} = \bigcap_{c_i \in N} D_i(c_i; \mu_{N_1})$ its center.
\end{lemma}

\begin{proof}
Let $\mu_{N_1}$ be the \v{C}ech scale of the disk system $N_1$ and let $c_{N_1}$ be the intersection point of the $\mu_M$-rescaled disk system. Then the point $c_{N_1}$ belongs to every disk $D_i(c_i; \mu_{N_1}\cdot 1)$, i.e. $\Vert c_i - c_{N_1} \Vert \leq \mu_{N_1}$ for any $c_i \in N$, thus $N \subset D(c_{N_1}; \mu_{N_1})$. On the other hand, by definition of \v{C}ech scale, $\mu_{N_1}$ is the minimal radius (scale) with such property. Therefore, by uniqueness, $D(c_{N_1}; \mu_{N_1})$ must be the minimal ball enclosing the point cloud $N$.
\end{proof}

In particular, for a point cloud $N$ in the plane we can apply our algorithm \texttt{Cech.scale} (Algorithm \ref{Algorithm:Cech-scale-triplets}) to the 2-disk system $N_1$, and get the minimal enclosing ball of $N$. However, the \v{C}ech scale of an arbitrary 2-disk system $M=\{D_1(c_1;r_1), \ldots, D_m(c_m;r_m)\}$ cannot be obtained from the minimal enclosing ball data of the point cloud $M_0=\{c_1, \ldots, c_m\}$.

In Figure \ref{Figure:Miniball} we show a point cloud $N$ (black dots) and the 2-disk system $N_1$ (blue circles). Applying the \texttt{Cech.scale} script to $N_1$ we get the \v{C}ech scale $\mu_{N_1}$ (radio of the red circle) and the point $c_{N_1}$ (red point, center of the red circle).

\begin{figure}[hbt]
\centering
\includegraphics[width=0.35\textwidth]{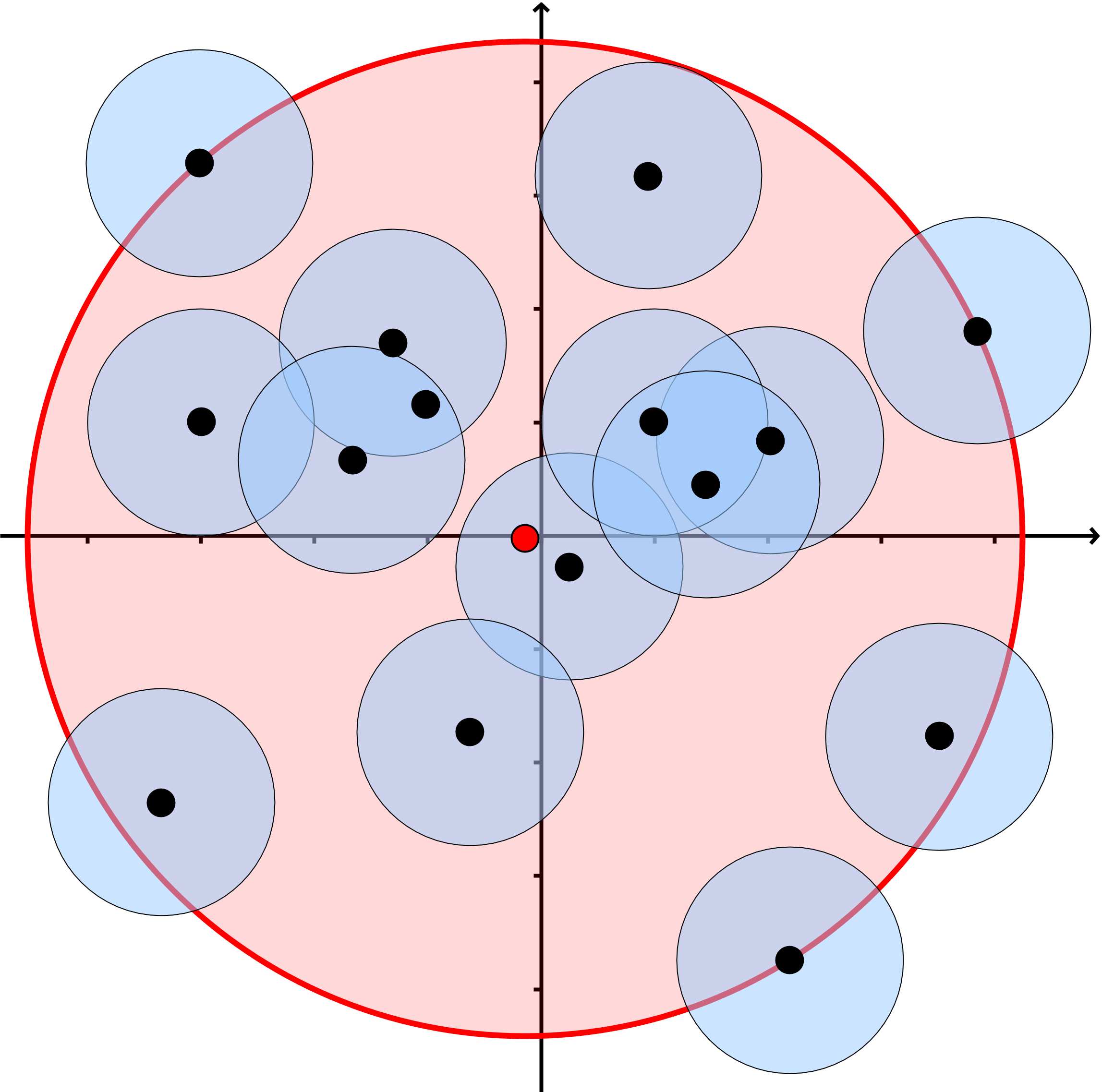}
\caption{The miniball of the point cloud $N$.}
\label{Figure:Miniball}
\end{figure}

For the miniball problem there are many efficient algorithms available online, which are easy to find. For example, the \texttt{C++} script in \cite{Fischer:Github} can compute the miniball for point clouds in any dimension (efficiently up to dimension 10,000). Such algorithms are not comparable with the \texttt{Cech.scale} algorithm because if only disk systems with equal (and unitary) radii were considered, several issues that were addressed in the case of different radii would be avoided.

\subsection{The algorithm {\normalfont \texttt{Cech.scale}} for higher dimensional disk systems}
It is not clear how to generalize the Algorithm \ref{Algorithm:Cech-scale-triplets} to determine the \v{C}ech scale of a disk system in $\mathbb{R}^d$ with $d>2$. However, it is possible to calculate the \v{C}ech scale if the $d$-disk system consists of only three disks. This makes it possible to calculate the 2-skeleton associated with a $d$-disk system in an arbitrary dimension. 

The relevance of this application lies in the possibility of calculating the 2-dimensional filtered simplicial \v{C}ech structure of a disk system immersed in a high-dimensional euclidean space. Many applications in topological data analysis concerns to the study of low dimensional topological features associated to a point data cloud immersed in a high-dimensional representation space. 

The key observation is that 
\[ \bigcap_{i = 1}^3 D(c_i; r_i) \neq \emptyset \Leftrightarrow \bigcap_{i = 1}^3 (D(c_i; r_i) \cap P) \neq \emptyset, \]
where $P$ is the affine plane generated by the set $\{c_1, c_2, c_3 \}$.

Thus, the problem of determining whether $\displaystyle \bigcap_{i = 1}^3 (D(c_i; r_i) \cap P)$ is empty or not, can be treated as one in the plane, constructing a disk system in $\mathbb R^2$ that preserves the affine configuration of the points $\{c_1, c_2, c_3\}$ in the affine space $P \subset \mathbb{R}^d$. 
To do this,  we set the first center $c_1$ as the origin in $\mathbb R^2$, and ``translate'' the others centers preserving their original configuration, taking care of moving the second center on the $x$-axis, as in Figure \ref{Figure:Affin-system}.

\begin{figure}[hbt]
\centering
\includegraphics[width=0.42\textwidth]{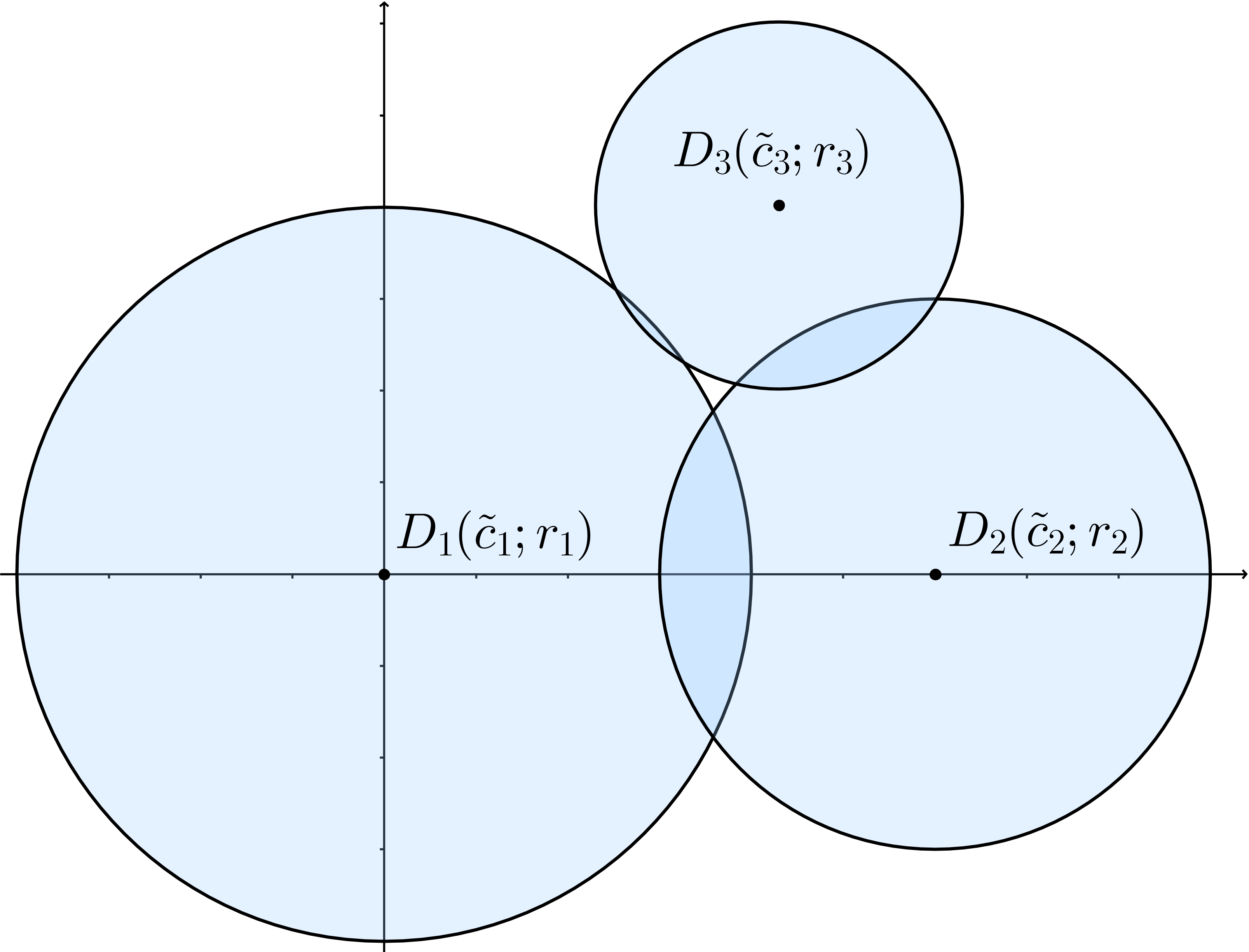}
\caption{Affin configuration of the $d$-disk system.}
\label{Figure:Affin-system}
\end{figure}

More precisely, to any $d$-disk system with three elements, say $M = \{D_1(c_1; r_1), D_2(c_2; r_2), D_3(c_3; r_3)\}$, we associate the following 2-disk system, which clearly preserves the affine configuration of the original centers $\{c_1,c_2,c_3\}$,
\[ \mathrm{Aff}(M) = \{D_1(\tilde{c}_1; r_1), D_2(\tilde{c}_2; r_2), D_3(\tilde{c}_3; r_3)\} \]
where $\tilde{c}_1 := (0,0)$, $\tilde{c}_2 := (\Vert c_2 - c_1 \Vert, 0)$, $\tilde{c}_3 := (\Vert c_3 - c_1 \Vert \cos(\theta), \Vert c_3 - c_1 \Vert \sin(\theta))$, where $\theta$ is the angle between the vectors $c_2 - c_1$ and $c_3 - c_1$, which satisfies the following relationship: 
\[ \cos \theta = \dfrac{\langle c_2 - c_1, c_3 - c_1 \rangle}{\Vert c_2 - c_1 \Vert \cdot \Vert c_3 - c_1 \Vert}.\]

The next algorithm is a variant of Algorithm \ref{Algorithm:Cech-weight-function}, taking as input a $d$-disk system in $\mathbb{R}^d$ and a nonnegative parameter $\lambda$, and as output the \v{C}ech weight function of the 2-skeleton of the generalized \v{C}ech complex structure. The algorithm first preprocess each triplet of $d$-disks as a 2-disk system, then the \v{C}ech scale is calculated.

\begin{algorithm} \label{Algorithm:2-skeletal-Cech-filtration}
    \SetKwInOut{Input}{Input}
    \SetKwInOut{Output}{Output}
    \Input{A $d$-disk system $M$ and a parameter $\lambda \geq 0$.}
    \Output{The \v{C}ech-weight function $\omega : \mathscr{C}_M(\lambda)^{(2)} \to \mathbb{R}$.}
    Calculate $\mathscr{C}_M(\lambda)^{(1)}$\;
   	\For{$\{D_i, D_j\} \in \mathscr{C}_M(\lambda)^{(1)}$}{
   		$\mathrm{LN}_{ij} \leftarrow $ $\lambda$-\texttt{LowerNbrs}$(D_i)\ \cap \ \lambda$-\texttt{LowerNbrs}$(D_j)$\;
   		\For{$D_k \in \mathrm{LN}_{ij}$}{
    		$N \leftarrow \mathrm{Aff}(\{D_i(c_i; r_i), D_j(c_j; r_j), D_k(c_k; r_k)\})$\;
			Calculate $\omega(N) \leftarrow \mu_N$\;
			\If{$\omega(N) \leq \lambda$}{
				Update $\mathscr{C}_M(\lambda)^{(2)} \leftarrow \mathscr{C}_M(\lambda)^{(2)} \cup \{D_i, D_j, D_k\}$\;
			}
   		}
   	}
	\Return($\mathscr{C}_M(\lambda)^{(2)}, \omega : \mathscr{C}_M(\lambda)^{(2)} \to \mathbb{R}$)
    \caption{2-skeletal \v{C}ech-weight function.}
\end{algorithm}

Figure \ref{Graph:Cech.scale.triplet-preprocessed.system} shows the performance (in $10^{-6}$ seconds) of the \texttt{C/C++} script \texttt{Cech.scale} (available in \cite{GCS}) and the preprocessing of the $d$-disk system to a 2-disk system.
% \newpage

\begin{figure}[hbt]
\centering
\includegraphics[width=0.9\textwidth]{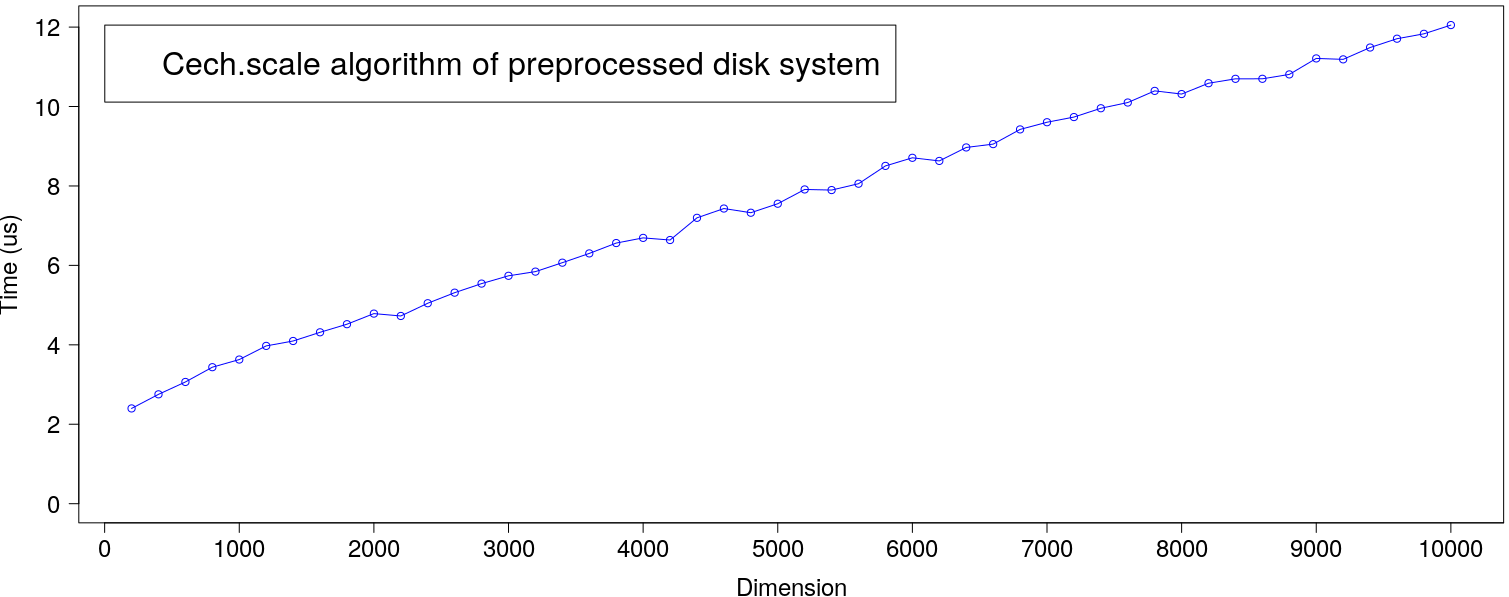}
\caption{Average time ($\mu$s) in high dimensions of the \texttt{Cech.scale} script and preprocessing disk systems.}
\label{Graph:Cech.scale.triplet-preprocessed.system}
\end{figure}

\begin{remark} \label{Remark:computer}\normalfont
All our timings were done on a 64-bit GNU/Linux machine with two Intel Xeon processors (3.40 GHz), although our script were not threaded and only one core was used per process. We measured all the timings with \texttt{clock()} from the Standard C library. The average times in both graphics (Figure \ref{Graph:Cech.scale-vs-triplets} and Figure \ref{Graph:Cech.scale.triplet-preprocessed.system}) are the mean times for $10^4$ repetitions of each algorithm, for every number of disks multiple of 10 in Figure \ref{Graph:Cech.scale-vs-triplets} from 10 to 500, and for every dimension multiple of 200 in Figure \ref{Graph:Cech.scale.triplet-preprocessed.system} from 200 to 10000.
\end{remark}

\section*{Acknowledgements}
The author J.-F. Espinoza acknowledges the financial support of PRODEP and of the Universidad de Sonora, as well as the ACARUS (High Performance Computing Area) for the support in the access to the clusters.

\bibliographystyle{abbrv}

\end{document}